\documentclass[a4paper,12pt]{article}
\usepackage[utf8]{inputenc}
\usepackage[T2A]{fontenc}
\usepackage{indentfirst}
\usepackage{amsmath}
\usepackage{amsthm}
\usepackage{amsfonts}
\usepackage{amssymb}
\usepackage{graphicx}
\usepackage{setspace}
\usepackage{array}
\usepackage{float}
\usepackage{url}
\usepackage{multirow}
\usepackage{multicol}
\usepackage{enumitem}
\usepackage[title]{appendix}
\usepackage{titling}
\usepackage{longtable}

\usepackage{cite}

\usepackage{booktabs}
\usepackage{comment}
\usepackage{mathtools}
\usepackage{hyperref}

\usepackage[left=2cm, right=2cm, top=2cm, bottom=2cm]{geometry}
\newcommand{\iprod}[2]{{\langle {#1},{#2} \rangle}}
\newcommand{\z}[1]{\mathbb{F}_{2}^{#1}}
\newcommand{\MF}[1]{\mathcal{M}_{#1}}
\newcommand{\MFC}[1]{\mathcal{M}_{#1}^{\#}}
\newcommand{\MFSP}[1]{\widehat{\mathcal{M}}_{#1}}
\newcommand{\MFCSP}[1]{\widehat{\mathcal{M}}_{#1}^{\#}}
\newcommand{\ClsC}[1]{#1^{\#}}
\newcommand{\ClsSP}[1]{\widehat{#1}}
\newcommand{\ClsCSP}[1]{\widehat{#1}^{\#}}
\newcommand{\subsp}[2]{\mathcal{L}_{#1}({#2})}
\newcommand{\linsp}[1]{[#1]}
\newcommand{\charfunc}[1]{\mathcal{\mathrm{Ind}}_{#1}}
\newcommand{\snkp}[2]{{\mathcal{S}_{#1}^{#2}}}
\newcommand{\asnk}[2]{{\mathcal{A}_{#1, #2}}}
\newcommand{\snk}[2]{\linsp{\asnk{#1}{#2}}}
\newcommand{\lbound}[1]{{\ell_{#1}}}
\newcommand{\Span}[1]{{\langle #1 \rangle}}
\newcommand{\near}[1]{\mathcal{N}(#1)}
\newcommand{\infomaps}[0]{\mathcal{I}}
\newcommand{\nearav}[1]{E({#1})}
\newcommand{\nbent}[4]{f_{#1, #2}^{#3, #4}}
\newcommand{\perk}[2]{#1_{#2}}
\newcommand{\coords}[2]{#1^{#2}}
\newcommand{\fper}[1]{\mathcal{P}_{#1}}
\newcommand{\sav}[2]{\rho_{#1,#2}}
\newcommand{\mfcbound}[1]{\vartheta_{2n}}
\newcommand{\sconstr}[5]{\mathcal{U}_{#1}^{#5}(#2, #3, #4)}
\newcommand{\bentset}[1]{\mathcal{B}_{#1}}
\newcommand{\xspace}[1]{\Gamma_{#1}}

\newtheorem{theorem}{Theorem}
\newtheorem{proposition}{Proposition}
\newtheorem{lemma}{Lemma}
\newtheorem{corollary}{Corollary}
\theoremstyle{remark}%
\newtheorem{remark}{Remark}%

\usepackage{array}

\allowdisplaybreaks

\providecommand{\keywords}[1]{
	\medskip
    \noindent
	\small\textbf{Keywords: } #1
	}

\providecommand{\affili}[1]{
	\begin{center} 
	\small #1
	\end{center}
	\vspace{.1cm}}

\title{On the Maiorana--McFarland Class Extensions}
\author{\normalfont Nikolay Kolomeec, Denis Bykov
\medskip
}
\date{}

\begin{document}
\large
\maketitle
\par\vspace{-60pt}
 \affili{
	Novosibirsk State University, Novosibirsk, Russia\\

    \medskip

 	\texttt{nkolomeec@gmail.com, den.bykov.2000i@gmail.com}
}
\begin{abstract}
       The closure $\MFC{m}$ and the extension $\MFSP{m}$ of the Maiorana--McFarland class $\MF{m}$ in $m = 2n$ variables relative to the extended-affine equivalence and the bent function construction $f \oplus \charfunc{U}$ are considered, where $U$ is an affine subspace of $\z{m}$ of dimension $m/2$. 
       We obtain an explicit formula for $|\MFSP{m}|$ and an upper bound for $|\MFCSP{m}|$. Asymptotically tight bounds for $|\MFC{m}|$ are proved as well, 
       for instance, $|\MFC{8}| \approx 2^{77.865}$. Metric properties of $\MF{m}$ and $\MFC{m}$ are also investigated. 
       We find the number of all closest bent functions to the set $\MF{m}$ and provide an upper bound of the same number for $\MFC{m}$. The average number $\nearav{\MF{m}}$ of $m/2$-dimensional affine subspaces of $\z{m}$ such that a function from $\MF{m}$ is affine on each of them is calculated. 
       We obtain that similarly defined $\nearav{\MFC{m}}$ satisfies $\nearav{\MFC{m}} < \nearav{\MF{m}}$ and $\nearav{\MFC{m}} = \nearav{\MF{m}} - o(1)$. 


\end{abstract}

  \section{Introduction}
  
  Bent functions are Boolean functions having interesting applications in cryptography, coding theory, algebra, etc. They are maximal nonlinear Boolean functions in an even number of variables. Many books 
  such as~\cite{
  TokarevaBook, MesnagerBook, LogachevEtAlBook, CusickStanica2017, CarletBook} are dedicated to them or contain information on them; some problems from student olympiads are dedicated to them as well~\cite{NSUCRYPTO2}. The investigations of bent functions were started in the 1960s both in the USA and the USSR~\cite{TokarevaBook}, the term appeared in~\cite{Rothaus1976}. At the same time, there are many open problems in this area. For instance, the number of all bent functions in $m = 2n$ variables is unknown if $m \geq 10$. 

There are important primary subclasses of bent functions such as the Maiorana--McFarland class $\MF{m}$~\cite{McFarland1973} containing bent functions 
$$
    f(x,y) = \iprod{x}{\pi(y)} \oplus \varphi(y), \ x, y \in \z{\frac{m}{2}},
$$ 
where $\pi$ is a permutation on $\z{\frac{m}{2}}$ and $\varphi: \z{\frac{m}{2}} \to \z{}$. Also, there are two most famous secondary constructions that can generate additional bent functions using some given class $\mathcal{K}$.
\begin{itemize}
	\item The extended-affine equivalence, which generates $\ClsC{\mathcal{K}}$ containing all
	$
		f(xA \oplus a) \oplus h(x)
	$
	for each $f \in \mathcal{K}$ in $m$ variables, invertible $m \times m$ matrix $A$ over $\z{}$, $a \in \z{m}$ and affine $h : \z{m} \to \z{}$.
	\item The construction from~\cite{Dillon1974}, which generates $\ClsSP{\mathcal{K}}$ containing 
	\begin{equation}\label{eq:subspaceContruction}
		f \oplus \charfunc{U}
	\end{equation}
	for each $f \in \mathcal{K}$ in $m$ variables and $\frac{m}{2}$-dimensional affine subspace $U \subseteq \z{m}$ such that $f$ is affine on $U$.
\end{itemize}
However, the cardinalities of $\ClsC{\mathcal{K}}$, $\ClsSP{\mathcal{K}}$ and $\ClsCSP{\mathcal{K}} = \ClsSP{\ClsC{\mathcal{K}}}$ are unknown for most of significant classes, and the calculation of them looks unfeasible. 
This work demonstrates that at least for $\MF{m}$ we can make progress in this direction.  
We obtain the exact value of $|\MFSP{m}|$, prove asymptotically tight bounds for $|\MFC{m}|$ and propose an upper bound for $|\MFCSP{m}|$, for instance,
\begin{equation}\label{eq:introMFCAsympt}
    |\MFC{m}| = \frac{(2^{m} - 1) \cdot \ldots \cdot (2^{\frac{m}{2} + 1} - 1)}{(2^{\frac{m}{2}} - 1) \cdot \ldots \cdot (2^{1} - 1)} |\MF{m}| - o(2^{\frac{m}{2}}!),
\end{equation} 
\begin{equation}\label{eq:introMFCSPBound}
    |\MFSP{m}| = (\frac{1}{18} 2^{m} + \frac{430}{63}) |\MF{m}| + o(|\MF{m}|)
\end{equation} 
and $|\MFCSP{m}| < (\frac{4}{3}2^{m} + 30) |\MFC{m}|$ for $m \geq 10$.
Note that the proved estimations are much more precise than the trivial
$$
  |\MFC{m}| \leq (2^{m} - 2^0) \cdot \ldots \cdot (2^{m} - 2^{m - 1}) |\MF{m}| \text{ and }  
$$ 
$$
  |\MFSP{m}| \leq 2^{\frac{m}{2}}(2^1 + 1) \cdot \ldots \cdot (2^{\frac{m}{2}} + 1) |\MF{m}|
$$ 
that follows from~\cite{Kolomeec2017}.
The class $\MF{m}$ is very important due to the simplicity of its functions. Its cardinality $2^{2^{\frac{m}{2}}}2^{\frac{m}{2}}!$ was often used as a lower bound for the number of all bent functions prior to work~\cite{PotapovTaranenkoTarannikov2024}. 
The class $\MFC{m}$ is called \textit{the completed Maiorana--McFarland} class. A natural question addressed to constructions is to generate bent functions outside~$\MFC{m}$, see, for instance,~\cite{Carlet1994, ZhangEtAl2017, ZhangEtAl2020, KudinEtAl2022_1, KudinEtAl2022_2, BapicEtAl2022, PasalicEtAl2023, PolujanThesis2021}.
There is the $\mathcal{D}$ class~\cite{Carlet1994} constructed using~(\ref{eq:subspaceContruction}) and $f$ from a subset of $\MF{m}$, i.e. $\mathcal{D} \subset \MFSP{m}$. 
In addition, there are other approaches to the obtained results.

First, all bent functions $\near{f}$ that are generated using the construction~(\ref{eq:subspaceContruction}) and some given bent function $f$ in $m$ variables are exactly all bent functions at the Hamming distance~$2^{\frac{m}{2}}$ from $f$~\cite{KolomeecPavlov2009}. It is the minimum possible distance between two bent functions. 
Moreover, $\near{f}$ are all closest bent functions to $f \in \MFC{m}$. We can define the set $\near{\MF{m}}$ as 
$$
    \near{\MF{m}} = \bigcup_{f \in \MF{m}} \near{f} \setminus \MF{m} = \MFSP{m} \setminus \MF{m}
$$
that consists of all closest bent functions to the set $\MF{m}$. Similarly, we determine $\near{\MFC{m}} = \MFCSP{m} \setminus \MFC{m}$ which is the set of all closest bent functions to $\MFC{m}$. Their cardinalities are
$$
    |\MFSP{m}| = |\near{\MF{m}}| + |\MF{m}| \text{ and } |\MFCSP{m}| = |\near{\MFC{m}}| + |\MFC{m}|.
$$
Thus, we also find the exact value of $|\near{\MF{m}}|$ and the upper bound for $|\near{\MFC{m}}|$, which are interesting metric properties of $\MF{m}$ and $\MFC{m}$. 

Secondly, any function $f \in \MF{m}$ is constructed as the concatenation of $2^{\frac{m}{2}}$ affine functions 
$
    x \in \z{\frac{m}{2}} \mapsto \iprod{\pi(y)}{x} \oplus \varphi(y)
$ which are distinct up to adding a constant, i.e. $f$ is affine on the affine subspaces $\z{\frac{m}{2}} \times \{y\}$ of $\z{m}$, where $y \in \z{\frac{m}{2}}$.
A natural combinatorial problem arises: what is the number of such $\frac{m}{2}$-dimensional affine subspaces of $\z{m}$ in total? According to~(\ref{eq:subspaceContruction}), the answer is $|\near{f}|$. 
Note that 
$\frac{m}{2}$ is the maximum possible dimension of such subspaces,
see, for instance,~\cite{LogachevEtAlBook}. 
In this work, we obtain the expected value $\nearav{\MF{m}}$ of $|\near{f}|$ for a random $f \in \MF{m}$:
$$
   \nearav{\MF{m}} = \frac{1}{|\MF{m}|}\sum_{f \in \MF{m}} |\near{f}|.
$$
Similarly, we determine $\nearav{\MFC{m}}$ for $f \in \MFC{m}$.
We obtain the exact value of $\nearav{\MF{m}}$, which can be also expressed as
\begin{equation}\label{eq:introNearAvAsympt}
    \nearav{\MF{m}} = \frac{10}{3} 2^{m} - 2^{\frac{m}{2}} + \frac{176}{21} + o(1),
\end{equation}
and prove that $\nearav{\MFC{m}} < \nearav{\MF{m}}$ (their difference is negligible) allowing us to estimate~$|\MFCSP{m}|$.
Taking into account the principle of constructing~$f \in \MF{m}$, we can make the assumption that $\nearav{\mathcal{K}} < \nearav{\MF{m}}$ for any sufficiently large~$\mathcal{K}$.

The outline. In Sections~\ref{sec:definitions},  necessary definitions are given. 
Section~\ref{sec:criterion} contains the criterion from~\cite{BykovKolomeec2023} for describing $\near{f}$ for $f \in \MF{m}$. We rewrite it in a more convenient form for this work (Section~\ref{sec:criterionRe}, Theorem~\ref{th:criterionMMFnew}) and point out a simple subcase generating most of the bent functions from $\near{\MF{m}}$ (Section~\ref{sec:DimL2Descr}, Theorem~\ref{theorem:criterionMMFdim2}). 
Section~\ref{sec:nearestToMF} studies intersections of $\near{f}$ and $\near{g}$ for $f, g \in \MF{m}$ (Theorem~\ref{th:coincidenceMMF}), its results are used  for the calculation of $|\near{\MF{m}}|$ and are important in the context of the mentioned subcase. 

In Section~\ref{sec:nearestAverage}, the explicit formula for $\nearav{\MF{m}}$ is proved (Theorem~\ref{th:NearAverage}) as well as some properties of this number (Corollaries~\ref{cor:AverageNumbSimplified}, \ref{cor:AverageNumbUpperBound} and \ref{cor:NearAverageAsympt}), see, for instance~(\ref{eq:introNearAvAsympt}). 
Next, we give the explicit formula for $|\near{\MF{m}}|$ in Section~\ref{sec:nearestNumber} (Theorem~\ref{th:NearMMFNumb}). Its asymptotic and some estimations are also given (Corollaries~\ref{cor:NearSimplified} and \ref{cor:NearMMFAsympt}). The results imply the expression~(\ref{eq:introMFCSPBound}) for $|\MFSP{m}|$ (Remark~\ref{remark:MFSPPower}).  

In Section~\ref{sec:mfcdescr}, we move to the $\MFC{m}$ class and prove some of its properties (Lemma~\ref{lemma:mfdescrequivalence} and Corollary~\ref{cor:fullCosetSerie}). 
Section~\ref{sec:BoundsMFC} proves the bounds for $|\MFC{m}|$ (Theorem~\ref{th:MFCLowerBound}, Corollary~\ref{cor:MFCLowerBoundAsympt} and Proposition~\ref{state:mcboundasympt}) that can be transformed to~(\ref{eq:introMFCAsympt}). For instance, $|\MFC{8}| \approx 2^{77.865}$ (see Table~\ref{table:mfcestimations}) that was estimated as at most $2^{81.38}$ in~\cite{LangevinLeander2011}.
Section~\ref{sec:nearestAverageMFC} contains an upper bound for $|\near{\MFC{m}}|$ obtained from $\nearav{\MFC{m}}$. The main result here is $\nearav{\MFC{m}} < \nearav{\MF{m}}$ for $m \geq 10$ and  $\nearav{\MFC{m}} = \nearav{\MF{m}} - o(1)$ (Theorem~\ref{th:MFCompletedNearAverage}). This implies the mentioned bound for $|\MFCSP{m}|$ (Corollary~\ref{cor:NearCompletedUpperBound} and Remark~\ref{remark:NearMFC}).


\section{Preliminaries}\label{sec:definitions}

\subsection{Boolean functions}
Let $\z{n} = \{ (x_1, \ldots, x_n): x_1, \ldots, x_n \in \z{} \}$
be the 
vector space over the field $\z{}$ consisting of two elements, the addition be denoted by $\oplus$ and 
$\iprod{x}{y} = x_1 y_1 \oplus \ldots \oplus x_n y_n$, 
$x, y \in \z{n}$.

\textit{A (vectorial) Boolean function} in $n$ variables is $f: \z{n} \to \z{}$ ($F: \z{n} \to \z{m}$), $\charfunc{S}$ is \textit{the characteristic Boolean function} of a set $S \subseteq \z{n}$.
\textit{The Hamming distance} between $f, g : \z{n} \to \z{}$ is equal to the number of $x \in \z{n}$ such that $f(x) \neq g(x)$.

$F: \z{n} \to \z{m}$ is called \textit{linear} if $F(x \oplus y) = F(x) \oplus F(y)$ for all $x, y \in \z{n}$.
It can be represented as $x \mapsto xA$ for some matrix $A$ over $\z{}$ of size $n \times m$.
By adding a constant from $\z{m}$ to linear functions, we get the set of all \textit{affine} functions.

Functions $f, g \in \z{n} \to \z{}$ are \textit{EA-equivalent} if $g = f \circ A \oplus h$, i.e. $g(x) = f(A(x)) \oplus h(x)$ for all $x \in \z{n}$, where $A: \z{n} \to \z{n}$ is affine and invertible, $h: \z{n} \to \z{}$ is affine.

\subsection{Subspaces and restrictions}

\textit{A linear subspace of} $\z{n}$ is a nonempty $L \subseteq \z{n}$ such that $x, y \in L$ implies $x \oplus y \in L$. 
The set $U = a \oplus L = \{a \oplus x : x \in L\}$ is called \textit{an affine subspace of} $\z{n}$, where $a \in \z{n}$. Also, $\linsp{U} = L = a \oplus U$. Their \textit{dimensions} are $\dim L = \log_2 |L|$ and $\dim U = \log_2 |U|$. 
The set of all $k$-dimensional affine (linear) subspaces of $\z{n}$ is denoted by $\asnk{n}{k}$ ($\snk{n}{k}$).
\textit{The orthogonal} $L^{\bot}$ to $L \in \snk{n}{k}$ is $\{ y \in \z{n} : \iprod{x}{y} = 0 \text{ for all } x \in L\}$ and belongs to $\snk{n}{n - k}$. The cardinality of $\snk{n}{k}$ is 
\begin{equation}\label{eq:snkpDefinition}
    \snkp{n}{k} = \prod_{i = 0}^{k - 1}\frac{2^{n} - 2^{i}}{2^{k} - 2^{i}} = \frac{(2^{n} - 1) \cdot \ldots \cdot (2^{n - k + 1} - 1)}{(2^{k} - 1) \cdot \ldots \cdot (2^{1} - 1)}
\end{equation}
for $0 \leq k \leq n$ and $0$ otherwise, $|\asnk{n}{k}| = 2^{n - k} \snkp{n}{k}$.

\textit{A restriction} of $F: \z{n} \to \z{m}$ to $S \subseteq \z{n}$ is denoted by $F|_S$. A function $H : U \to \z{m}$ is \textit{affine (linear)}, if $H = F|_{U}$ for some affine (linear) $F$ and $U$ is an affine (linear) subspace of~$\z{n}$. We will also say that $F$ is affine on $U$ if $F|_{U}$ is affine.

\subsection{Indices and information coordinates}

Let $I = \{i_1, \ldots, i_k\}$, where $1 \leq i_1 < i_2 < \ldots < i_k \leq n$.
The following notations for $x \in \z{n}$ and $y \in \z{k}$ are used:
\begin{itemize}
  \item $\perk{x}{I} = (x_{i_1}, \ldots, x_{i_k}) \in \z{k}$,
  \item $\coords{y}{I} = z \in \z{n}$, where $z_{i_1} = y_1, \ldots, z_{i_k} = y_k$ and $z_j = 0$ for $j \notin I$, i.e. $\perk{(\coords{y}{I})}{I} = y$ holds.
\end{itemize}
Also, $\overline{I} = \{1, \ldots, n\} \setminus I$. 
For $H: \z{n} \to \z{k}$, we denote by $\coords{H}{I}$ the function $x \mapsto \coords{H(x)}{I}$ and the same for $\perk{\pi}{I}$, where $\pi: \z{n} \to \z{n}$. If $I = \{i\}$, $\pi_i = \perk{\pi}{I}$ will be also used.

The set $I$ is called \textit{an information set of} $L \in \asnk{n}{k}$ if $$\{ (x_{i_1}, \ldots, x_{i_{k}}) : x \in L\} = \z{k}.$$
One can find some $I$ for any $L \in \asnk{n}{k}$ using Gaussian elimination for its basis matrix. It is well known that $I$ is an information set of $L$ $\iff$ $\overline{I}$ is an information set of $\linsp{L}^{\bot}$.
One more of its important properties is that any coset of $\linsp{L}$ can be uniquely determined as $\coords{a}{\overline{I}} \oplus L$, where $a \in \z{n - k}$. 

Hereinafter, $\infomaps$ is some arbitrary fixed mapping that for any $L \in \asnk{n}{k}$, $0 \leq k \leq n$, gives us its information set $I = \infomaps(L)$. We need it only to choose information sets deterministically. 

\subsection{Bent functions and the Maiorana--McFarland class}

A function $f: \z{2n} \to \z{}$ is \textit{a bent function} if it is at the maximum possible Hamming distance from the set of all affine Boolean functions in $2n$ variables; they form the set $\bentset{2n}$. Though $m = 2n$ is used in the introduction (which makes the dependence of the estimates on the number of variables more clear), we use bent functions in $2n$ variables elsewhere.

\textit{The Maiorana--McFarland class} $\MF{2n}$ consists of 
\begin{equation*}
    f_{\pi, \varphi}(x, y) = \iprod{x}{\pi(y)} \oplus \varphi(y),
\end{equation*} 
where $\pi: \z{n} \to \z{n}$ is invertible and $\varphi: \z{n} \to \z{}$. All such functions are bent functions.
Let us denote by $\fper{n}$ the set of all invertible $\pi: \z{n} \to \z{n}$ and define for each $\pi \in \fper{n}$
\begin{equation*}
\subsp{k}{\pi} = \{ U \in \asnk{n}{k} : \pi(U) = \{\pi(x) : x \in U\} \in \asnk{n}{k}\}.
\end{equation*}

Also, $\xspace{n} = \z{n} \times \{0 \in \z{n}\}$. We note that any $f_{\pi, \varphi} \in \MF{2n}$ is affine on each coset of $\xspace{n}$. 

\textit{The completed Maiorana--McFarland class} $\MFC{2n}$ is the closure of $\MF{2n}$ with respect to EA-equivalence. 
In fact, we can consider only linear transformations of coordinates since $f(x \oplus a) \oplus h(x) \in \MF{2n}$ for any $f \in \MF{2n}$, $a \in \z{2n}$ and affine $h: \z{2n} \to \z{}$.
The set of all bent functions $\bentset{2n}$ is closed with respect to EA-equivalence.

\section{The closest bent functions to a given bent function}\label{sec:criterion}

It is known~\cite[Corollary~3]{KolomeecPavlov2009} that all bent functions at the minimum possible distance $2^n$ from a given $f \in \bentset{2n}$ can be generated using the construction~(\ref{eq:subspaceContruction}).  \begin{proposition}[see~\cite{KolomeecPavlov2009}]\label{th:BasicCriterion}
  Let $f \in \bentset{2n}$ and $U \subset \z{2n}$, $|U| = 2^n$.
  Then $f \oplus \charfunc{U} \in \bentset{2n}$ $\iff$ $U \in \asnk{2n}{n}$ and $f|_U$ is affine.
\end{proposition}

In the case of $f \in \MFC{2n}$ there always exists a bent function at the distance $2^n$ from $f$, i.e. we can define the set
\begin{equation}
    \near{f} = \{ f \oplus \charfunc{U} : U \in \asnk{2n}{n} \text{ and } f|_{U} \text{ is affine }\}   
\end{equation}
and refer to it as \textit{the set of all closest bent functions to} $f$. Similarly, 
\begin{equation}
    \near{\mathcal{K}} = \bigcup_{f \in \mathcal{K}} \near{f} \setminus \mathcal{K}  
\end{equation}
is \textit{the set of all closest bent functions to the class} $\mathcal{K}$, where $\mathcal{K}$ is either $ \MF{2n}$ or $\MFC{2n}$. 
Also, $|\near{\MF{2n}}|$ and $|\near{\MFC{2n}}|$ are directly connected with $|\MFSP{2n}|$ and $|\MFCSP{2n}|$:
\begin{equation*}
    |\MFSP{2n}| = |\near{\MF{2n}}| + |\MF{2n}|,\ 
    |\MFCSP{2n}| = |\near{\MFC{2n}}| + |\MFC{2n}|
\end{equation*}
since $f \oplus \charfunc{\xspace{n}} \in \MF{2n}$ for any $f \in \MF{2n}$.

\subsection{The case of $f \in \MF{2n}$} \label{sec:criterionRe}

The criterion~\cite[Theorem 3]{BykovKolomeec2023} describes $\near {f}$ for any $f\in\MF{2n}$. We provide it swapping $x$ and $y$. Also, $\Span{M}$ is the linear span of the rows of a matrix $M$ over $\z{}$.

\begin{proposition}[see~\cite{BykovKolomeec2023}]\label{state:criterionMMFOld}
  Let $f_{\pi, \varphi}\in \MF{2n}$ and $U = (a,b) \oplus \Span{M}$ for $a, b \in \z{n}$ and $M$ over $\z{}$ represented as 
  \begin{equation}\label{eq:basisMatrixBlockForm}
  M =
  \begin{pmatrix}
    T & S \\
    P & 0
  \end{pmatrix},
  \end{equation}
  where $S$ and $P$ of size $k \times n$ and $(n-k) \times n$ have full rank, $0 \leq k \leq n$ and $T$ is arbitrary of size $k \times n$.
  Then $f_{\pi, \varphi}\oplus \charfunc{U} \in \bentset{2n}$ $\iff$ the following conditions are satisfied:
  \begin{enumerate}
    \item $\pi(b \oplus \Span{S}) = \pi(b) \oplus \Span{P}^{\bot}$,
    \item $u \in \z{k}\mapsto \iprod{uT \oplus a}{\pi(uS \oplus b)}\oplus \varphi (uS \oplus b)$ is affine.
  \end{enumerate}
\end{proposition}

Let us describe a construction allowing us to represent any $U \in \asnk{2n}{n}$ in the unique way.
\begin{proposition}\label{state:subspaceRepresentation}
    Any $U \in \asnk{2n}{n}$ can be uniquely expressed as 
    $$
        \sconstr{n}{L}{R}{H}{\infomaps} = \{ (\coords{H}{\infomaps(R^{\bot})}(y) \oplus z, y) : y \in L, z \in R\}, 
    $$
    where $L \in \asnk{n}{k}$, $R \in \snk{n}{n - k}$,  and $H: L \to \z{k}$ is affine, $k \in \{0, \ldots, n\}$. 
    Moreover, 
    $\sconstr{n}{L}{R}{H}{\infomaps}$ is linear $\iff$ $L$ and $H$ are linear. Also, $\linsp{\sconstr{n}{L}{R}{H}{\infomaps}} \cap \xspace{n} = R$.
\end{proposition}
\begin{proof}
    First of all, $\sconstr{n}{L}{R}{H}{\infomaps} \in \asnk{2n}{n}$. Indeed, let $L' = \linsp{L}$, $b \in L$ and $I = \infomaps(R^{\bot})$, i.e. $\overline{I}$ is some information set of $R$. By definition, there exist linear $A: \z{n} \to \z{k}$ and $a \in \z{k}$ such that $H = A|_L \oplus a$. Thus, 
    \begin{multline}\label{eq:affineSubspConstr}
         \sconstr{n}{L}{R}{H}{\infomaps} = \{(\coords{A}{I}(b \oplus y) \oplus \coords{a}{I} \oplus z, b \oplus y) : y \in L', z \in R\} \\
         = (\coords{A}{I}(b) \oplus \coords{a}{I}, b) \oplus \sconstr{n}{L'}{R}{A}{\infomaps}.
    \end{multline} 
    $U' = \sconstr{n}{L'}{R}{A}{\infomaps} \in \snk{2n}{n}$ holds since $L'$, $R$, $A|_{L'}$ are linear and $|U'| = 2^{\dim L} \cdot 2^{\dim R} = 2^n$. Also, $U' \cap \xspace{n} = R$. 
    The equality~(\ref{eq:affineSubspConstr}) gives us $2^{n - \dim R} \cdot 2^{n - \dim L} = 2^n$ cosets of $U'$ which  are distinct for distinct $(a, b \oplus L')$. 

    The representation is unique. 
    Indeed, distinct $(L', R)$ generate distinct $U'$. Also, if $A(y) \neq A'(y)$ for some $y \in L'$, then the cosets $\coords{A}{I}(y) \oplus R$ and $\coords{A'}{I}(y) \oplus R$ are distinct due to the choice of $I$. This implies    
    $\sconstr{n}{L'}{R}{A}{\infomaps} \neq \sconstr{n}{L'}{R}{A'}{\infomaps}$.

    Hence, we only need to prove that any $U' \in \snk{2n}{n}$ can be represented in this way. 
    Let us choose a basis matrix $M$ of $U$ in the form~(\ref{eq:basisMatrixBlockForm}) that can be obtained, for instance, from the reduced row echelon form after Gaussian elimination. Moreover, Gaussian elimination can transform $T$ to $\coords{W}{I}$ without changing $P$ and $S$, where $W$ is a matrix of size $k \times k$ and $I = \infomaps(\Span{P}^{\bot})$. This means that
    \begin{multline}\label{eq:blockToSubspaceRepr}
        U' = \{(\coords{(uW)}{I} \oplus vP, uS) : u \in \z{k}, v \in \z{n - k}\} \\= \{(\coords{(V(y)W)}{I} \oplus z, y) : y \in \Span{S}, z \in \Span{P}\},
    \end{multline}
    where $V: y \mapsto u$ such that $uS = y$. Since $V$ is linear, $U' = \sconstr{n}{\Span{S}}{\Span{P}}{W \circ V}{\infomaps}$. 
\end{proof}

Next, we rewrite the criterion in terms of Proposition~\ref{state:subspaceRepresentation}. 

\begin{theorem}\label{th:criterionMMFnew}
  Let $f_{\pi, \varphi} \in \MF{2n}$. Then $\near{f_{\pi, \varphi}}$ consists of all $f_{\pi, \varphi} \oplus \charfunc{\sconstr{n}{L}{\linsp{\pi(L)}^{\bot}}{H}{\infomaps}}$, 
  where
  \begin{enumerate}
    \item $L \in \subsp{k}{\pi}$ for $0 \leq k \leq n$, 
    \item both $H: L \to \z{k}$ and 
    $
      x \in L \mapsto \iprod{H(x)}{\perk{\pi}{I}(x)} \oplus \varphi(x)
    $ are affine, $I = \infomaps(\linsp{\pi(L)})$.
  \end{enumerate}
  Distinct $(L, H)$ correspond to distinct bent functions. 
\end{theorem}
\begin{proof}
    According to Proposition~\ref{state:subspaceRepresentation}, we can represent any $U \in \asnk{2n}{n}$ as $U = \sconstr{n}{b \oplus L}{R}{H}{\infomaps}$, where $L \in \snk{n}{k}$ and $R \in \snk{n}{n - k}$ whose basis matrices are $S$ and $P$, and $H(x) = xA \oplus c$ for some matrix $A$ of size $n \times k$. Since $x \in b \oplus L$ can be represented as $b \oplus uS$ and $z \in R$ as $vP$, where $u \in \z{k}$, $v \in \z{n - k}$, each element of $U$ can be represented as
    \begin{multline*}
        (\coords{H}{I}(uS \oplus b) \oplus vP, uS \oplus b) = (\coords{((uS \oplus b)A \oplus c)}{I} \oplus vP, b \oplus uS) \\
        = (b\coords{A}{I} \oplus \coords{c}{I}, b) \oplus (u\coords{(SA)}{I} \oplus vP, uS), \ u \in \z{k}, \ v \in \z{n - k}.
    \end{multline*}
    Therefore, one of basis matrices of $U$ is 
    $$
    \begin{pmatrix}
    \coords{(SA)}{I} & S \\
    P & 0
  \end{pmatrix}.    
    $$ 
    The first condition of Proposition~\ref{state:criterionMMFOld} gives us that $b \oplus L \in \subsp{k}{\pi}$ and $R = \linsp{\pi(b \oplus L)}^{\bot}$. The second condition is the affinity of
    \begin{multline*}
        u \in \z{k} \mapsto \iprod{u\coords{(SA)}{I} \oplus \coords{(bA \oplus c)}{I}}{\pi(uS \oplus b)}\oplus \varphi (uS \oplus b)  \\ 
        = \iprod{\coords{H}{I}(uS \oplus b)}{\pi(uS \oplus b)}\oplus \varphi (uS \oplus b)
    \end{multline*}
    which is equivalent to the affinity of the function $x \in b \oplus L \mapsto \iprod{H(x)}{\perk{\pi}{I}(x)}\oplus \varphi (x)$ due to $\coords{H}{I}_i \equiv 0$ for $i \in \overline{I}$. 
\end{proof}

For shortness, we will denote any bent function from $\near{f_{\pi, \varphi}}$ or $\near{\MF{2n}}$ by $\nbent{\pi}{\varphi}{L}{H}$ which is equal to
\begin{equation}\label{eq:nearbentDenot}
     f_{\pi, \varphi}(x, y) \oplus \charfunc{L}(y) \cdot \charfunc{\coords{H}{\infomaps(\linsp{\pi(L)})}(y) \oplus \linsp{\pi(L)}^{\bot}}(x),
\end{equation}
where $L$ and $H$ satisfy the conditions of Theorem~\ref{th:criterionMMFnew}.
Though $H$ is not defined on $\z{n} \setminus L$, $\charfunc{L}(y) = 0$ for any $y \notin L$.

We also note that the properties of $\pi \in \fper{n}$ related to its $\subsp{k}{\pi}$ were studied, for instance, in~\cite{ClarkEtAl20071, LiEtAl2020, KolomeecBykov2024, Kolomeec2024}. 

\subsection{The most important special cases}\label{sec:DimL2Descr}

Let us pay attention to the most important special cases of Theorem~\ref{th:criterionMMFnew}. We start with one described in~\cite[Corollary 2]{BykovKolomeec2023} that gives us the lower bound $\lbound{2n}$ for $|\near{f_{\pi, \varphi}}|$.
\begin{proposition}[see~\cite{BykovKolomeec2023}]\label{state:mcfarlandcase}
    Let $f_{\pi, \varphi} \in \MF{2n}$. Then $|\near{f_{\pi, \varphi}} \cap \MF{2n}| = \lbound{2n}$, where $\lbound{2n} = 2^{2n + 1} - 2^n$. Moreover, $\nbent{\pi}{\varphi}{L}{H} \in \MF{2n}$ $\iff$ $\dim L \leq 1$, where $\nbent{\pi}{\varphi}{L}{H} \in \near{f_{\pi, \varphi}}$, see~(\ref{eq:nearbentDenot}). 
\end{proposition}
One more description of $\near{f} \cap \MF{2n}$ for $f \in \MF{2n}$ can be found in~\cite[Proposition~10]{Kolomeec2017}. We also note that this bound is accurate for some $n$~\cite{BykovKolomeec2023, BykovKolomeec2025}.

Finally, we consider the case of $\dim L = 2$ which generates most of the bent functions from $\near{\MF{2n}}$, see Sections~\ref{sec:nearestAverage} and \ref{sec:nearestNumber}.
Let us recall the following well-known property.
\begin{lemma}\label{lemma:affinity2}
  Let $H: L \to U$, where $L$ and $U$ are affine subspaces of $\z{n}$ and $\z{k}$, $\dim L = 2$ and $\dim U \leq 2$.
  Then $H$ is affine $\iff$ $\bigoplus_{x \in L} H(x) = 0$.
  Any invertible $H$ is affine.
\end{lemma}

\begin{theorem}\label{theorem:criterionMMFdim2}
  Let $f_{\pi, \varphi} \in \MF{2n}$ and $L \in \subsp{2}{\pi}$.
  Let $a, b, c \in L$ such that $\perk{\pi}{I}(a) = (0,1)$, $\perk{\pi}{I}(b) = (1,0)$ and $\perk{\pi}{I}(c) = (1,1)$, where $I = \infomaps(\linsp{\pi(L)})$.
  Then $\nbent{\pi}{\varphi}{L}{H} \in \near{f_{\pi, \varphi}}$ $\iff$  
  \begin{enumerate}
    \item $H(a), H(b), H_1(c)$ are arbitrary 5 bits,
    \item $H_2(c) = H_2(a) \oplus H_1(b) \oplus H_1(c) \oplus \sum_{x \in L} \varphi(x)$,
    \item $H(a \oplus b \oplus c) = H(a) \oplus H(b) \oplus H(c)$.
  \end{enumerate}
  These give us $|\near{f_{\pi, \varphi}}| \geq \lbound{2n} + 2^5|\subsp{2}{\pi}|$.
\end{theorem}
\begin{proof}
  We only need to check the second condition of Theorem~\ref{th:criterionMMFnew}.
  Let $H: L \to \z{2}$ is arbitrary.
  According to Lemma~\ref{lemma:affinity2}, the function
  $
    x \in L \mapsto \iprod{H(x)}{\perk{\pi}{I}(x)} \oplus \varphi(x)
  $
  is affine $\iff$
  \begin{equation*}
    0 = \bigoplus_{x \in L} \left(\iprod{H(x)}{\perk{\pi}{I}(x)} \oplus \varphi(x)\right) 
    = \bigoplus_{x \in L} \varphi(x) \oplus H_2(a) \oplus H_1(b) \oplus H_1(c) \oplus H_2(c).
  \end{equation*}
  This is satisfied$\iff$ the second point is true.
  At the same time, the third point $H(a \oplus b \oplus c) = H(a) \oplus H(b) \oplus H(c)$ is satisfied $\iff$ $H$ is affine.
\end{proof}

\section{Intersections of $\near{f}$ and $\near{g}$ for $f, g \in \MF{2n}$}\label{sec:nearestToMF}

The calculation of the cardinalities of $\MFSP{2n}$ and $\near{\MF{2n}}$ requires information related to intersections of $\near{f}$ and $\near{g}$ for $f, g \in \MF{2n}$. This is the most technically complex section, but it contains, however, a briefly formulated result (see Remark~\ref{remark:mmfBentCoincidence}) which will be used only in Section~\ref{sec:nearestNumber}.
We start with the following lemma.   

\begin{lemma}\label{lemma:piPhiOutsideL}
	Let $\nbent{\pi}{\varphi}{L}{H}, \nbent{\pi'}{\varphi'}{L'}{H'} \in \near{\MF{2n}}$ and $\nbent{\pi}{\varphi}{L}{H} = \nbent{\pi'}{\varphi'}{L'}{H'}$. Then $L = L'$, $\pi|_{\z{n} \setminus L} = \pi'|_{\z{n} \setminus L}$ and $\varphi|_{\z{n} \setminus L} = \varphi'|_{\z{n} \setminus L}$. Also, $\dim L \geq 3$ implies $(\pi, \varphi, H) = (\pi', \varphi', H')$.
\end{lemma}
\begin{proof} Since $\nbent{\pi}{\varphi}{L}{H}, \nbent{\pi'}{\varphi'}{L'}{H'} \notin \MF{2n}$, Proposition~\ref{state:mcfarlandcase} claims that $\dim L \geq 2$ and $\dim L' \geq 2$. Let $I = \infomaps(\linsp{\pi(L)})$ and $J = \infomaps(\linsp{\pi'(L')})$
Then 
(\ref{eq:nearbentDenot}) gives us that
    \begin{multline}\label{eq:nearfuncequality}
        \iprod{x}{\pi(y)} \oplus \varphi(y) \oplus \charfunc{L}(y) \, \charfunc{\coords{H}{I}(y) \oplus \linsp{\pi(L)}^{\bot}}(x) = \\
        \iprod{x}{\pi'(y)} \oplus \varphi'(y) \oplus \charfunc{L'}(y) \, \charfunc{\coords{H'}{J}(y) \oplus \linsp{\pi'(L')}^{\bot}}(x).
    \end{multline}
    Let us fix some $y \in \z{n}$. The function on the left is affine $\iff$ $y \in L$. Indeed, $\dim \linsp{\pi(L)}^{\bot} = n - \dim \pi(L) \leq n - 2$, i.e. $\charfunc{\coords{H}{I}(y) \oplus \linsp{\pi(L)}^{\bot}}$ cannot be affine. Similarly, the function on the right is affine $\iff$ $y \in L'$. 
    Hence, (\ref{eq:nearfuncequality}) implies $L = L'$.
	
	Fixing some $y \in \z{n} \setminus L$ in~(\ref{eq:nearfuncequality}), we obtain
	\begin{equation*}
		\iprod{x}{\pi(y)} \oplus \varphi(y)  = \iprod{x}{\pi'(y)} \oplus \varphi'(y) \text{ for all } x \in \z{n},
	\end{equation*}
	i.e. $\pi|_{\z{n} \setminus L} = \pi'|_{\z{n} \setminus L}$, $\varphi|_{\z{n} \setminus L} = \varphi'|_{\z{n} \setminus L}$
    and $\pi(L) = \pi'(L')$. 
    
    Let $R = \linsp{\pi(L)}^{\bot}$ and, again, fix some $y \in L$ in~(\ref{eq:nearfuncequality}):
    \begin{equation}\label{eq:nearBentCoincidenceForL}
        \iprod{x}{\pi(y) \oplus \pi'(y)} \oplus \varphi(y) \oplus \varphi'(y) 
        =  \charfunc{\coords{H}{I}(y) \oplus R}(x) \oplus \charfunc{\coords{H'}{I}(y) \oplus R}(x).
    \end{equation}
    Suppose that it is satisfied for $\dim L > 2$. The function on the left is affine. But the function on the right is affine only if the cosets $\coords{H}{I}(y) \oplus R$ and $\coords{H'}{I}(y) \oplus R$ coincide for all $y \in L$. Otherwise, the function on the right will take $1$ exactly $2^{n - \dim L + 1} < 2^{n - 1}$ times for some $y \in L$, i.e. it cannot be affine. 
    However, the coincidence of the cosets for all $y \in L$ implies that the function on the left must be identically zero for all $y \in L$, i.e. $\pi = \pi'$ and $\varphi = \varphi'$. Finally, Theorem~\ref{th:criterionMMFnew} guarantees that $\nbent{\pi}{\varphi}{L}{H} \neq \nbent{\pi}{\varphi}{L}{H'}$ if $H \neq H'$.
\end{proof}

Thus, we only need to consider the case of Theorem~\ref{theorem:criterionMMFdim2}.

\begin{lemma}\label{lemma:dimL2Auxiliary}
	Let $L \in \asnk{n}{2}$, $\varphi: L \to \z{}$, $\sigma, \sigma' : L \to \z{2}$ be invertible and both $H: L \to \z{2}$ and $h: x \in L \mapsto \iprod{H(x)}{\sigma(x)} \oplus \varphi(x)$ be affine. 	
	Let us define $\varphi': L \to \z{}$ as
	$$
		\varphi'(x) = \varphi(x) \oplus \iprod{H(x)}{\sigma(x) \oplus \sigma'(x)} \oplus \charfunc{\{\sigma(x)\}}(\sigma'(x)) \oplus 1,
	$$
	$x \in L$, and $H': L \to \z{2}$ in the following way:  
		\begin{enumerate}
			\item $H'(x) = H(x)$ $\iff$ $\sigma'(x) = \sigma(x)$ for all $x \in L$,
			\item $\iprod{H'(x) \oplus H(x)}{\sigma'(x) \oplus \sigma(x)} = 0$ for all $x \in L$.
		\end{enumerate}
	Then $H'$ is uniquely determined and affine. Moreover, $h': x \in L \mapsto \iprod{H'(x)}{\sigma'(x)} \oplus \varphi'(x)$ is affine as well.
\end{lemma}
\begin{proof}
	Let $\delta(x) = \sigma(x) \oplus \sigma'(x)$ and $T(x) = H(x) \oplus H'(x)$, $x \in L$. According to the conditions, $T(x) = 0$ $\iff$ $\delta(x) = 0$ and $\iprod{T(x)}{\delta(x)} = 0$ for all $x \in L$. Moreover, $T$ is uniquely determined $\iff$ $H'$ is uniquely determined; and $T$ is affine $\iff$ $H'$ is affine.

	It is clear that $T$ is uniquely determined. Indeed, for $\delta(x) \neq 0$ the equation $\iprod{u}{\delta(x)} = 0$, $u \in \z{2}$, has exactly two solutions $\{0, \delta(x)\}^{\bot}$ and one of them is $u = 0$. It means that $T(x)$ must be equal to the nonzero solution due to $T(x) \neq 0$ for this case. If $\delta(x) = 0$, then $T(x) = 0$ as well.
	
	Let us show that $T$ is affine. According to Lemma~\ref{lemma:affinity2}, we need to prove that
	\begin{equation}\label{eq:affinitysum}
		\bigoplus_{x \in L} T(x) = 0.
	\end{equation} 
	Since $\sigma$ and $\sigma'$ are invertible, they are affine due to Lemma~\ref{lemma:affinity2}. This means that $\delta$ is affine as their sum.
	Hence, either $\delta$ is invertible or the number of solutions of $\delta(x) = b$ is even for any $b \in \z{2}$.
	
	If $\delta$ is invertible, then $T$ is invertible as well. Indeed, if $\delta(x) = 0$, then $T(x) = 0$. 
    Otherwise $T(x)$ is the only nonzero element of $\{0, \delta(x)\}^{\bot}$; they are distinct for distinct $\delta(x)$.
    Thus, $T$ is affine due to Lemma~\ref{lemma:affinity2}.
	
	If $\delta$ is not invertible, then $\delta(x) = b$ has an even number of solutions (or does not have solutions). Also, if $\delta(x_1) = \delta(x_2) = b$, $x_1, x_2 \in L$, then $T(x_1) = T(x_2)$ and $T(x_1) \oplus T(x_2) = 0$ as well. 
    Hence,~(\ref{eq:affinitysum}) is satisfied. 
	
	Let us prove that the function $h'$ is affine. 
    First, the equality
    $$
        \iprod{H(x)}{\delta(x)} \oplus \iprod{H'(x)}{\sigma'(x)} = \iprod{H(x)}{\sigma(x)}  \oplus \iprod{T(x)}{\sigma'(x)} 
    $$
    holds since both its parts are equal to $\iprod{H(x)}{\sigma(x)} \oplus \iprod{H(x)}{\sigma'(x)} \oplus \iprod{H'(x)}{\sigma'(x)}$.
    Secondly, using $\varphi'(x) = \varphi(x) \oplus \iprod{H(x)}{\delta(x)} \oplus \charfunc{\{0\}}(\delta(x)) \oplus 1$ and the equality above, we obtain that $h'(x) \oplus 1$ is equal to     
	\begin{multline*}
		\iprod{H'(x)}{\sigma'(x)} \oplus \varphi(x) \oplus \iprod{H(x)}{\delta(x)} \oplus \charfunc{\{0\}}(\delta(x)) \oplus 1 \oplus 1\\
        = \iprod{H(x)}{\sigma(x)} \oplus \varphi(x) \oplus \iprod{T(x)}{\sigma'(x)} \oplus \charfunc{\{0\}}(\delta(x))\\     
		 = h(x) \oplus \iprod{T(x)}{\sigma'(x)} \oplus \charfunc{\{0\}}(\delta(x)).
	\end{multline*}	
	Since $h$ is affine, it is enough to prove that the function $r: x \in L \mapsto \iprod{T(x)}{\sigma'(x)} \oplus \charfunc{\{0\}}(\delta(x))$ is affine. According to Lemma~\ref{lemma:affinity2}, $r$ is affine $\iff$ it takes $0$ an even number of times. Let us calculate this number.
	
	If $\delta(x) = 0$, then $T(x) = 0$ and $r(x) = \iprod{T(x)}{\sigma'(x)} \oplus \charfunc{\{0\}}(\delta(x)) = 1$, which is not interesting for us. 
	
	Let $\delta(x) \neq 0$. Then $T(x) \neq 0$ and $\charfunc{\{0\}}(\delta(x)) = 0$, i.e. $r(x) = \iprod{T(x)}{\sigma'(x)}$. Note that  
	\begin{equation}\label{eq:sigmasigmaprime}
		\iprod{T(x)}{\sigma'(x)} = \iprod{T(x)}{\sigma(x)}
	\end{equation}	
	due to $\iprod{T(x)}{\delta(x)} = 0$. Thus, $\iprod{T(x)}{\sigma(x)} = 0$ $\iff$ $\iprod{T(x)}{\sigma'(x)} = 0$ . 
	But the equation $\iprod{T(x)}{y} = 0$ has exactly two solutions: some nonzero $y \in \z{2}$ and $0$. Since $\sigma(x) \neq \sigma'(x)$ in this case,  either ($\sigma(x) = y$, $\sigma'(x) = 0$) or ($\sigma(x) = 0$, $\sigma'(x) = y$). Moreover, we cannot have more than two such $x$ since $\sigma$ and $\sigma'$ are invertible.
		
	Let $\sigma(x) = y$ and $\sigma'(x) = 0$.
	Since $\sigma$ and $\sigma'$ are invertible, there always exists an additional $x' \in L$, $x' \neq x$, such that $\sigma(x') = 0$ and $\sigma'(x') \neq 0$, i.e. $\delta(x') \neq 0$. 
    We can use~(\ref{eq:sigmasigmaprime}) for $x'$ and obtain $r(x') = \iprod{T(x')}{\sigma'(x')} = 0$. 
    
    If $\sigma(x) = 0$ and $\sigma'(x) = y$, we immediately obtain some $x' \neq x$ with $\sigma'(x') = 0 \neq \sigma(x')$, i.e. $\delta(x') \neq 0$ and $r(x') = 0$.
    
    Thus, either there are exactly two $x, x' \in L$ such that $r(x) = r(x') = 0$ or $r(x) \equiv 1$. Consequently, $r$ and $h'$ are affine.
\end{proof}

Let us union the proved properties into one theorem. 
\begin{theorem}\label{th:coincidenceMMF}
	Let $\nbent{\pi}{\varphi}{L}{H} \in \near{\MF{2n}}$ and $I = \infomaps(\linsp{\pi(L)})$. Then 
	\begin{enumerate}
    	\item The case of $\dim L \leq 1$ is impossible.
		\item If $\dim L = 2$, $\nbent{\pi}{\varphi}{L}{H} \in \near{f_{\pi', \varphi'}}$ for exactly 24 distinct $f_{\pi', \varphi'} \in \MF{2n}$. Moreover, $\nbent{\pi}{\varphi}{L}{H} = \nbent{\pi'}{\varphi'}{L'}{H'}$ $\iff$
		\begin{enumerate}
			\item $L' = L$, $\pi'|_{\z{n} \setminus L} = \pi|_{\z{n} \setminus L}$ and $\varphi'|_{\z{n} \setminus L} = \varphi|_{\z{n} \setminus L}$,
            \item $\varphi'|_{L}$ and $H': L \to \z{2}$ are uniquely determined as
$$
    \varphi'(y) = \varphi(y) \oplus \iprod{H(y)}{\perk{\pi}{I}(y) \oplus \perk{\pi'}{I}(y)} 
        \oplus \charfunc{\{\pi(y)\}}(\pi'(y)) \oplus 1,
$$
    \begin{multline*}
    \iprod{H'(y) \oplus H(y)}{\perk{\pi}{I}(y) \oplus \perk{\pi'}{I}(y)} = 0 \\  \text{ and } 
    H'(y) = H(y) \iff \pi'(y) = \pi(y).
    \end{multline*}

		\end{enumerate}
		\item If $\dim L \geq 3$, $\nbent{\pi}{\varphi}{L}{H} \in \near{f_{\pi',\varphi'}}$ $\iff$ $\pi = \pi'$ and $\varphi = \varphi'$.
	\end{enumerate}
\end{theorem}
\begin{proof}
	The first point follows from Proposition~\ref{state:mcfarlandcase}. 
    Due to Lemma~\ref{lemma:piPhiOutsideL}, the third point is satisfied and $\nbent{\pi}{\varphi}{L}{H} = \nbent{\pi'}{\varphi'}{L'}{H'}$ requires $L = L'$, $\pi'|_{\z{n} \setminus L} = \pi|_{\z{n} \setminus L}$ and $\varphi'|_{\z{n} \setminus L} = \varphi|_{\z{n} \setminus L}$. 
    
    Let us prove the case of $\dim L = 2$.	    
	According to the proof of Lemma~\ref{lemma:piPhiOutsideL}, $\nbent{\pi}{\varphi}{L}{H} = \nbent{\pi'}{\varphi'}{L}{H'}$, where $\pi'|_{\z{n} \setminus L} = \pi|_{\z{n} \setminus L}$ and $\varphi'|_{\z{n} \setminus L} = \varphi|_{\z{n} \setminus L}$ $\iff$ (\ref{eq:nearBentCoincidenceForL}) is satisfied for any $y \in L$:
    \begin{equation}\label{eq:nearBentCoincidenceForL2}
        \iprod{x}{\pi(y) \oplus \pi'(y)} \oplus \varphi(y) \oplus \varphi'(y) 
        = \charfunc{\coords{H}{I}(y) \oplus R}(x) \oplus \charfunc{\coords{H'}{I}(y) \oplus R}(x),
    \end{equation}
    where $R = \linsp{\pi(L)}^{\bot} = \linsp{\pi'(L)}^{\bot}$ and the second condition of Theorem~\ref{th:criterionMMFnew} is satisfied for affine $H': L \to \z{2}$. 
	
    Next, let us choose any possible $\pi'|_L$. We note that the coincidence of the cosets from the indicators implies that $H(y) = H'(y)$ since $\overline{I}$ is some information set of $R$, see Section~\ref{sec:definitions}. 
    Assuming that~(\ref{eq:nearBentCoincidenceForL2}) is satisfied, $H(y) = H'(y)$ implies $\pi(y) = \pi'(y)$ and $\varphi(y) = \varphi'(y)$, where $y$ is some element of $L$. Moreover, $\pi(y) = \pi'(y)$ implies $H(y) = H'(y)$ due to $\dim R = n - 2$. The formulas from the condition provide the same since $I$ is an information set of $\pi(L)$. 
	
	Let us consider $y \in L$ such that $\pi(y) \neq \pi'(y)$, i.e. $H(y) \neq H'(y)$ is also required. First, we point out that $\pi(y) \oplus \pi'(y) \in \linsp{\pi(L)}$. Secondly, the choice of $I$ provides that any $x \in \z{n}$ can be uniquely represented as $\coords{z}{I} \oplus x'$, where $x' \in R$ and $z \in \z{2}$. Thus,
	\begin{multline*}
		\iprod{x}{\pi(y) \oplus \pi'(y)} = \iprod{\coords{z}{I} \oplus x'}{\pi(y) \oplus \pi'(y)} = \\
		\iprod{\coords{z}{I}}{\pi(y) \oplus \pi'(y)} \oplus 0 = \iprod{z}{\perk{\pi}{I}(y) \oplus \perk{\pi'}{I}(y)}.
	\end{multline*}
    Similarly, $\charfunc{\coords{H}{I}(y) \oplus R}(x)$ transforms to $\charfunc{\{H(y)\}}(z)$. Consequently, (\ref{eq:nearBentCoincidenceForL2}) is equivalent to
    $$
        \iprod{z}{\perk{\pi}{I}(y) \oplus \perk{\pi'}{I}(y)} \oplus \varphi(y) \oplus \varphi'(y) = \charfunc{\{H(y), H'(y)\}}(z),
    $$
    where $z \in \z{2}$. Since $\iprod{z}{\perk{\pi}{I}(y) \oplus \perk{\pi'}{I}(y)} \oplus \varphi(y) = 1$ has exactly two solutions, (\ref{eq:nearBentCoincidenceForL2}) is satisfied $\iff$ these solutions are $z = H(y)$ and $z = H'(y)$, i.e. (\ref{eq:nearBentCoincidenceForL2}) transforms to
	$$
		\begin{cases}
		\iprod{H(y)}{\perk{\pi}{I}(y) \oplus \perk{\pi'}{I}(y)} \oplus \varphi(y) \oplus \varphi'(y) = 1,\\	
		\iprod{H'(y)}{\perk{\pi}{I}(y) \oplus \perk{\pi'}{I}(y)} \oplus \varphi(y) \oplus \varphi'(y) = 1.\\	
		\end{cases}
	$$
	This system is equivalent to the following:
	$$
		\begin{cases}
		\iprod{H(y) \oplus H'(y)}{\perk{\pi}{I}(y) \oplus \perk{\pi'}{I}(y)} = 0, \\	
		\varphi'(y) = \varphi(y) \oplus \iprod{H(y)}{\perk{\pi}{I}(y) \oplus \perk{\pi'}{I}(y)} \oplus 1.
		\end{cases}
	$$
	Since $H(y) \neq H'(y)$ and $\pi(y) \neq \pi'(y)$ in this case, the formulas from the statement of the theorem are obtained. 
    
    At the same time, we still need to prove that $H'$ is uniquely determined and $\nbent{\pi'}{\varphi'}{L}{H'}$ belongs to $\near{f_{\pi', \varphi'}}$, i.e. the obtained $H'$ satisfies the second condition of Theorem~\ref{th:criterionMMFnew}. 
    However, all conditions of Lemma~\ref{lemma:dimL2Auxiliary} are satisfied for $\sigma = \perk{\pi}{I}|_L$, $\sigma' = \perk{\pi'}{I}|_L$ and $H$. Moreover, $\pi(y) = \pi'(y)$ $\iff$ $\sigma(y) = \sigma'(y)$ due to the choice of $I$, where $y \in L$. Hence, Lemma~\ref{lemma:dimL2Auxiliary} gives us that both $H'$ and $y \in L \mapsto \iprod{H'(y)}{\pi'(y)} \oplus \varphi'(y)$ are affine and $H'$ is uniquely determined, i.e. $\nbent{\pi'}{\varphi'}{L}{H'} \in \near{f_{\pi', \varphi'}}$ by Theorem~\ref{th:criterionMMFnew}. 
    Since we choose any $\pi'|_L$ with $\pi'|_{\z{n} \setminus L} = \pi|_{\z{n} \setminus L}$ and then uniquely determine $\varphi'$ and $H'$, there are exactly $4! = 24$ distinct $f_{\pi', \varphi'}$ such that $\nbent{\pi}{\varphi}{L}{H} \in \near{f_{\pi', \varphi'}}$.
\end{proof}

\begin{remark}\label{remark:mmfBentCoincidence}
Theorem~\ref{th:coincidenceMMF} constructs all $\nbent{\pi'}{\varphi'}{L'}{H'} $ that are equal to some given $\nbent{\pi}{\varphi}{L}{H} \in \near{\MF{2n}}$ using $L' = L$ of dimension $2$ and all 24 distinct $\pi'$ such that $\pi'|_{\z{n} \setminus L} = \pi|_{\z{n} \setminus L}$ . After that the given formulas uniquely determine $\varphi'$ and $H'$.  
\end{remark}

\section{The expected value of $|\near{f}|$ for $f \in \MF{2n}$}\label{sec:nearestAverage}

Let us introduce 
$$
    \nearav{\MF{2n}} = \frac{1}{|\MF{2n}|} \sum_{f \in \MF{2n}} |\near{f}|
$$
which is the expected value of the number of $U \in \asnk{2n}{n}$ such that a random $f \in \MF{2n}$ is affine on $U$. 
We also need to define
\begin{equation}\label{eq:rhoDefinition}
    \sav{n}{k} = 2^{2(n - k)}(\snkp{n}{k})^2 2^k! (2^n - 2^k)! / 2^n!.
\end{equation}
The most important values are
	\begin{align}
		\sav{n}{2} &= 
        \frac{2^{2n - 3}}{3} + \frac{1}{12} + \frac{1}{2^{n + 2} - 12}, \ n \geq 2,  \label{eq:rhoDefinition2}\\
		\sav{n}{3} &= \frac{5}{224} \frac{2^{n}(2^n - 1)(2^n - 2)(2^n - 4)}{(2^n - 3)(2^n - 5)(2^n - 6)(2^n - 7)}, n \geq 3. \label{eq:rhoDefinition3}
	\end{align}
Let us start with auxiliary properties.

\subsection{Auxiliary results}

The following result related to the expected value of $|\subsp{k}{\pi}|$ for $\pi \in \fper{n}$, was proved in~\cite[Proposition 5]{Kolomeec2024}. For  completeness, we give it with the proof which is not difficult.  
\begin{lemma}[see~\cite{Kolomeec2024}]\label{lemma:sumPi}
	 Let $0 \leq k \leq n$. Then
	$$
		\sum_{\pi \in \fper{n}} |\subsp{k}{\pi}| = 2^n!\,\sav{n}{k}.
	$$
\end{lemma}
\begin{proof}
Let $\delta_{\pi, L} = 1$ if $\pi(L) \in \asnk{n}{k}$ and $0$ otherwise, where $\pi \in \fper{n}$ and $L \in \asnk{n}{k}$. Let us rewrite the sum above in the following way:
\begin{multline}\label{eq:sumPi}
	\sum_{\pi \in \fper{n}} |\subsp{k}{\pi}| = \sum_{\pi \in \fper{n}} \sum_{L \in \asnk{n}{k}} \delta_{\pi, L} = \sum\limits_{L \in \asnk{n}{k}} \sum_{\pi \in \fper{n}} \delta_{\pi, L} \\= \sum_{L \in \asnk{n}{k}} |\{\pi \in \fper{n} : \pi(L) \in \asnk{n}{k}\}|.
\end{multline}
Let us choose any $L \in \asnk{n}{k}$ and calculate the number of distinct $\pi \in \fper{n}$ such that $\pi(L) \in \asnk{n}{k}$. First, we can choose any $k$-dimensional affine subspace of $\z{n}$ as $\pi(L)$. The same $\pi$ cannot have distinct $\pi(L)$, i.e. we can choose $\pi(L)$ in $2^{n - k} \snkp{n}{k}$ ways. Secondly, $\pi|_L$ can be chosen in $|\pi(L)|! = 2^k!$ ways since $\pi(L)$ is already defined. Each such $\pi|_L$ can be extended to the whole $\z{n}$ in $(2^n - 2^k)!$ ways since $\pi(\z{n}\setminus L) = \z{n} \setminus \pi(L)$. Finally, taking into account that the initial $L$ can be chosen in $2^{n - k} \snkp{n}{k}$ ways and~(\ref{eq:sumPi}), we obtain that
$$
	\sum_{\pi \in \fper{n}} |\subsp{k}{\pi}| = 2^{n - k} \cdot \snkp{n}{k} \cdot 2^{n - k} \cdot \snkp{n}{k} \cdot 2^k! \cdot (2^n - 2^k)!.
$$	
This completes the proof.
\end{proof}

A quite easy one is~\cite[Proposition~7]{Kolomeec2024}, we provide it in a little bit simplified form. 
\begin{lemma}[see~\cite{Kolomeec2024}]\label{lemma:averageSumLim}
	It holds that $\sav{n}{k} > \sav{n + 1}{k}$, where $n \geq k$. Also,
		$\lim_{n \to \infty} \sum_{k = 4}^{n - 1} \sav{n}{k} = 0.$
\end{lemma}

In addition to properties of permutations, we prove the one related to the second condition of Theorem~\ref{th:criterionMMFnew}.
\begin{lemma}\label{lemma:sumPhiH}
	Let $\pi: L \to \z{k}$ be invertible, where $L \in \asnk{n}{k}$. Let the set $\mathcal{H}_{\pi, \varphi}$ for $\varphi: \z{n} \to \z{}$ consist of all affine $H: L \to \z{k}$ such that the function $x \in L \mapsto \iprod{\pi(x)}{H(x)} \oplus \varphi(x)$ is affine as well. Then
	$$
		\sum_{\varphi : \z{n} \to \z{}} |\mathcal{H}_{\pi, \varphi}| = 2^{2^n - 2^k + (k + 1)^2}.
	$$
\end{lemma}
\begin{proof}
	The  sum from the statement of the lemma is equal to
	$
		2^{2^n - 2^k} \sum_{\varphi|_L : L \to \z{}} |\mathcal{H}_{\pi, \varphi}|
	$
	since any $\varphi|_L$ can be extended to $\varphi$ in $2^{2^n - 2^k}$ ways.
	Next,
	$$
		\sum_{\varphi|_L : L \to \z{}} |\mathcal{H}_{\pi, \varphi}| = \sum_{\varphi|_L : L \to \z{}} \sum_{H} \delta_{\varphi, H} = \sum_{H} \sum_{\varphi|_L : L \to \z{}}  \delta_{\varphi, H},
	$$	 
	where $\delta_{\varphi, H} = 1$ if $x \in L \mapsto \iprod{\pi(x)}{H(x)} \oplus \varphi(x)$ is affine and $0$ otherwise.
	Next, it is not difficult to see that for each fixed pair $\pi, H$ the restriction $\delta_{\varphi, H} = 1$ is satisfied exactly for $2^{k + 1}$ distinct $\varphi|_L$ since $\varphi|_L$ is defined up to its affine part. The number of distinct affine $H: L \to \z{k}$ is $2^{k^2 + k}$. 
\end{proof}

\subsection{The expression for $\nearav{\MF{2n}}$}

Now we are ready to find the expression for $\nearav{\MF{2n}}$.
\begin{theorem}\label{th:NearAverage}
	$
	   \nearav{\MF{2n}} = \lbound{2n} + \sum_{k = 2}^{n} \sav{n}{k} \cdot 2^{(k + 1)^2 - 2^k}.
	$
\end{theorem}  
\begin{proof}
    Theorem~\ref{th:criterionMMFnew} gives us that
  \begin{equation}\label{eq:expressionForSum}
    |\near{f_{\pi, \varphi}}| = \sum_{k = 0}^{n} \sum_{L \in \subsp{k}{\pi}} |\mathcal{H}_{\pi_L, \varphi}|,
  \end{equation}
    where $\pi_L = \perk{\pi}{\infomaps(\linsp{\pi(L)})}$ and $\mathcal{H}_{\pi_L, \varphi}$ is defined in Lemma~\ref{lemma:sumPhiH} (L\ref{lemma:sumPhiH}) according to the second condition of Theorem~\ref{th:criterionMMFnew}.
    Thus, 
    \begin{align}
        \sum_{f_{\pi, \varphi} \in \MF{2n}} \sum_{L \in \subsp{k}{\pi}} |\mathcal{H}_{\pi_L, \varphi}| &= \sum_{\pi \in \fper{n}} \sum_{\varphi: \z{n} \to \z{}} \sum_{L \in \subsp{k}{\pi}} |\mathcal{H}_{\pi_L, \varphi}|
        \nonumber \\ &=  \sum_{\pi \in \fper{n}} \sum_{L \in \subsp{k}{\pi}} \sum_{\varphi: \z{n} \to \z{}} |\mathcal{H}_{\pi_L, \varphi}| \nonumber \\
        &\stackrel{\text{L\ref{lemma:sumPhiH}}}{=} \sum_{\pi \in \fper{n}} \sum_{L \in \subsp{k}{\pi}} 2^{2^n - 2^k + (k + 1)^2} \nonumber \\
        &= 2^{2^n - 2^k + (k + 1)^2} \sum_{\pi \in \fper{n}} |\subsp{k}{\pi}| \nonumber \\
        &\stackrel{\text{L\ref{lemma:sumPi}}}{=} 2^{2^n} 2^n! \, 2^{(k + 1)^2 - 2^k} \sav{n}{k}. \label{eq:kCorrespondence}
    \end{align}
    According to~(\ref{eq:expressionForSum}), $k = 0$ and $k = 1$ are the cases of Proposition~\ref{state:mcfarlandcase} generating exactly $\lbound{2n}$ bent functions from $\MF{2n}$. Thus, (\ref{eq:expressionForSum}) and (\ref{eq:kCorrespondence}) provide that $\nearav{\MF{2n}}$ is equal to
    \begin{equation*}
       \sum_{f_{\pi, \varphi} \in \MF{2n}} \frac{|\near{f_{\pi, \varphi}}|}{|\MF{2n}|}  = \lbound{2n} + \sum_{k = 2}^{n} 2^{(k + 1)^2 - 2^k} \sav{n}{k}.
    \end{equation*}
    We also note that $k = 2$ is the case of  Theorem~\ref{theorem:criterionMMFdim2}.
\end{proof}
The equality~(\ref{eq:rhoDefinition2}) implies the following.
\begin{corollary}\label{cor:AverageNumbSimplified}
		 $\nearav{\MF{2n}} = \frac{10}{3} 2^{2n} - 2^n + \frac{8}{3} + \frac{8}{2^{n} - 3} + 2^{(n + 1)^2 - 2^n} + \sum_{k = 3}^{n - 1} 2^{(k + 1)^2 - 2^k} \sav{n}{k}$, where $n \geq 2$.
    Rounded values of the last sum can be found in Table~\ref{table:approximationSum}.
\end{corollary}

\begin{table}[!t]
\renewcommand{\arraystretch}{1.3}
\centering
\begin{tabular}{c|cc|c}
\hline
$2n$ & $\sav{n}{3} \cdot 2^8$ & $\sav{n}{4} \cdot 2^9$ & $\sum_{k = 3}^{n - 1} \sav{n}{k} \cdot 2^{(k + 1)^2 - 2^k}$ \\
\hline
8 & 17.902098 & 512 & 17.9020979020979034 \\
10 & 9.355732 & 0.0032743174336464 & 9.3590067076487955 \\
12 & 7.203453 & 0.0000071066250978 & 7.2034599713804699 \\
14 & 6.394514 & 0.0000000489712760 & 6.3945141111150132 \\
16 & 6.040081 & 0.0000000005242934 & 6.0400809074314523 \\
18 & 5.873799 & 0.0000000000068280 & 5.8737987726531991 \\
20 & 5.793219 & 0.0000000000000976 & 5.7932190779699999 \\
22 & 5.753549 & 0.0000000000000015 & 5.7535493917318199 \\
24 & 5.733867 & 0.0000000000000000 & 5.7338671451801089 \\
\hline
\end{tabular}
\caption{The values for estimating $\nearav{\MF{2n}}$ and $|\near{\MF{2n}}|$}
\label{table:approximationSum}
\end{table}

\begin{corollary}\label{cor:AverageNumbUpperBound}
	$
		 \nearav{\MF{2n}} < \frac{10}{3} 2^{2n} - 2^n + 29
	$ 
    for $n \geq 5$. 
\end{corollary}
\begin{proof}
    By Corollary~\ref{cor:AverageNumbSimplified}: $\frac{8}{3} + \frac{8}{2^n - 3} < 3$ and $2^{(n + 1)^2 - 2^n} \leq 16$ if $n \geq 5$. The sum $\sum_{k = 3}^{n - 1} 2^{(k + 1)^2 - 2^k} \sav{n}{k}$ is less than $10$, see Table~\ref{table:approximationSum} and Lemma~\ref{lemma:averageSumLim}.
\end{proof}
    
\begin{corollary}\label{cor:NearAverageAsympt}
    $
        \nearav{\MF{2n}} = \frac{10}{3} 2^{2n} - 2^n + \frac{176}{21} + o(1).
    $ 
    It can also be written as
    $
        \nearav{\MF{2n}} = \lbound{2n} + \frac{4}{3} 2^{2n} + \frac{176}{21} + o(1)$ or $\nearav{\MF{2n}} = \frac{5}{3} \lbound{2n} + \frac{2}{3} 2^n +\frac{176}{21} + o(1).
    $
\end{corollary}
\begin{proof}
    Let us note that
    $$
        \sum_{k = 4}^{n - 1}  2^{(k + 1)^2 - 2^k} \sav{n}{k} \leq 2^9 \sum_{k = 4}^{n - 1} \sav{n}{k},
    $$
    since $2^{(k + 1)^2 - 2^k}$ is equal to $2^9$ for $k = 4$, $2^{4}$ for $k = 5$ and is less than $1$ for $k \geq 6$. Thus, Lemma~\ref{lemma:averageSumLim} guarantees that 
    \begin{equation}\label{eq:partsFrom4ton}
      \sum_{k = 4}^{n - 1} 2^{(k + 1)^2 - 2^k} \sav{n}{k} = o(1).
    \end{equation}
    By Theorem~\ref{th:NearAverage}, the equality above, (\ref{eq:rhoDefinition2}) and~(\ref{eq:rhoDefinition3}), $\nearav{\MF{2n}}$ is 
    \begin{multline*}
         \lbound{2n} + 2^5 \sav{n}{2} + 2^8 \sav{n}{3} + \sum_{k = 4}^{n - 1} 2^{(k + 1)^2 - 2^k} \sav{n}{k} + 2^{(n + 1)^2 - 2^n} \\
        = 2^{2n + 1} - 2^n + 2^5 \big(\frac{2^{2n - 3}}{3} + \frac{1}{12} \big) + 2^8 \cdot \frac{5}{224} + o(1), 
    \end{multline*}
    which is equal to $\frac{10}{3} 2^{2n} - 2^n + \frac{8}{3} + \frac{40}{7} + o(1)$ as well as to all equalities from the statement of the corollary.
\end{proof}

Thus, a random $f \in \MF{2n}$, which is constructed using $2^n$ distinct (up to adding a constant) affine functions in $n$ variables, is affine on approximately $\frac{5}{3} \lbound{2n}$ or $\frac{10}{3} 2^{2n}$ distinct $U \in \asnk{2n}{n}$. Taking into account the lower bound $|\near{f}| \geq \lbound{2n}$ (see Proposition~\ref{state:mcfarlandcase}), 
the expected value of $|\near{f}|$ is asymptotically $\frac{5}{3}$ times greater than the minimum possible value of $|\near{f}|$. 
Indeed, there are bent functions $f \in \MF{2n}$ with $|\near{f}| = \lbound{2n} + o(\lbound{2n})$~\cite{BykovKolomeec2025}. Moreover, the bound $\lbound{2n}$ is accurate for some $n$~\cite{BykovKolomeec2023, BykovKolomeec2025}. As we will show in Section~\ref{sec:nearestAverageMFC}, the situation is exactly the same for $f \in \MFC{2n}$ due to $\nearav{\MFC{2n}} = \nearav{\MF{2n}} - o(1)$ (see Theorem~\ref{th:MFCompletedNearAverage}).

Also, we point out that the number of variables here is $m = 2n$. Hence, the construction~(\ref{eq:subspaceContruction}) for some random $f \in \MF{2n}$ generates slightly more bent functions than the trivial one that just adds to the given bent function one of $2^{2n + 1}$ affine functions. At the same time, it can be considered as a generalization of the construction~(\ref{eq:subspaceContruction}) for affine subspaces of $\z{2n}$ of arbitrary dimension~\cite{Carlet1994}. 
It would be interesting to know how many bent functions this generalized construction generates, since its properties are also connected with the affinity~\cite{Kolomeec2021}.

\section{The cardinalities of $\near{\MF{2n}}$ and $\MFSP{2n}$}\label{sec:nearestNumber}

Here we find the explicit expressions for $|\near{\MF{2n}}|$ and $|\MFSP{2n}|$ using the results of Sections~\ref{sec:nearestToMF} and \ref{sec:nearestAverage}. 

\begin{theorem}\label{th:NearMMFNumb}
	$|\near{\MF{2}}| = 0$ and for $n \geq 2$ the following holds:
	$$
		|\near{\MF{2n}}| = \big(\frac{4}{3} \sav{n}{2} + \sum_{k = 3}^{n} \sav{n}{k} \cdot 2^{(k + 1)^2 - 2^k} \big)|\MF{2n}|.
	$$
    Moreover, the case of Theorem~\ref{theorem:criterionMMFdim2} generates $\frac{4}{3} \sav{n}{2} |\MF{2n}|$ bent functions from $\near{\MF{2n}}$.
\end{theorem}
\begin{proof}
    Similarly to the proof of Theorem~\ref{th:NearAverage}, (\ref{eq:expressionForSum}), Proposition~\ref{state:mcfarlandcase} and Theorem~\ref{th:coincidenceMMF} give us that $|\near{\MF{2n}}|$ is equal to
    \begin{equation*}
        \frac{1}{24} \sum_{f_{\pi, \varphi} \in \MF{2n}} \sum_{L \in \subsp{2}{\pi}} |\mathcal{H}_{\pi_L, \varphi}|  +   \sum_{k = 3}^{n} \sum_{f_{\pi, \varphi} \in \MF{2n}} \sum_{L \in \subsp{k}{\pi}} |\mathcal{H}_{\pi_L, \varphi}|.
    \end{equation*}
    Indeed, the functions generated in the case of $k \in \{0, 1\}$ belong to $\MF{2n}$. The case of $k \geq 3$ generates the functions belonging to only one $\near{f_{\pi, \varphi}}$, see Theorem~\ref{th:coincidenceMMF}. Finally, if $k = 2$, each function is calculated exactly $24$ times in the sum $\sum_{f_{\pi, \varphi} \in \MF{2n}} \sum_{L \in \subsp{2}{\pi}} |\mathcal{H}_{\pi_L, \varphi}|$.
	The equality~(\ref{eq:kCorrespondence}) completes the proof.
\end{proof}

\begin{corollary}\label{cor:NearSimplified}
		$\frac{|\near{\MF{2n}}|}{|\MF{2n}|} = \frac{1}{18}2^{2n} + \frac{1}{9} + \frac{1}{3 \cdot 2^{n} - 9} + 2^{(n + 1)^2 - 2^n} + \sum_{k = 3}^{n - 1} 2^{(k + 1)^2 - 2^k} \sav{n}{k}$, where $n \geq 2$. See also Table~\ref{table:approximationSum}.
\end{corollary}

\begin{corollary}\label{cor:NearMMFAsympt}
	$
		|\near{\MF{2n}}| = \big(\frac{1}{18} 2^{2n} +  \frac{367}{63}\big)|\MF{2n}| + o(|\MF{2n}|).
	$
\end{corollary}
\begin{proof}
    Since~(\ref{eq:partsFrom4ton}) holds, Corollary~\ref{cor:NearSimplified} and (\ref{eq:rhoDefinition3}) give us that 
    $$
        \frac{|\near{\MF{2n}}|}{|\MF{2n}|} = \frac{2^{2n}}{18} + \frac{1}{9} + 2^8 \sav{n}{3} + o(1) = \frac{2^{2n}}{18} + \frac{1}{9} + \frac{40}{7} + o(1)
    $$
    that is exactly what we need.
\end{proof}
\begin{remark}
    Note that $|\near{\MF{2n}}| > \big(\frac{1}{18} 2^{2n} + \frac{367}{63}\big)|\MF{2n}|$ for $n \geq 2$. Also, $(\frac{1}{18}2^{2n} + \frac{1}{9} + \frac{1}{3 \cdot 2^{n} - 9}) |\MF{2n}|$ bent functions can be generated using Theorem~\ref{theorem:criterionMMFdim2}.
\end{remark}
\begin{remark}\label{remark:MFSPPower}
    Since $|\MFSP{2n}| = |\near{\MF{2n}}| + |\MF{2n}|$, Theorem~\ref{th:NearMMFNumb} and Corollaries~\ref{cor:NearSimplified}, \ref{cor:NearMMFAsympt} give us expressions for $|\MFSP{2n}|$ as well. For instance, $|\MFSP{2n}| = \big(\frac{1}{18} 2^{2n} +  \frac{430}{63}\big)|\MF{2n}| + o(|\MF{2n}|)$.
\end{remark}

Table~\ref{table:NMFandMFSPPower} provides $|\near{\MF{2n}}|$ and $|\MFSP{2n}|$ for small $n$. Interestingly, $\MFSP{2} = \bentset{2}$, $\MFSP{4} = \bentset{4}$ and almost half of bent functions from $\bentset{6}$ belong to $\MFSP{6}$.
\begin{table}[!t]
\renewcommand{\arraystretch}{1.3}
\centering
\begin{tabular}{ccccc}
\hline
$2n$ & $|\MF{2n}|$ & $|\near{\MF{2n}}|$ & $|\MFSP{2n}|$ & $|\bentset{2n}|$\\
\hline
2 & 8 & 0 & 8 & 8\\
4 & 384 & 512 & 896 & 896\\
6 & $\approx 2^{23.299}$ & $\approx 2^{31.320}$ & $\approx 2^{31.326}$ & $\approx 2^{32.337}$\\
8 & $\approx 2^{60.250}$ & $\approx 2^{69.338}$ & $\approx 2^{69.341}$ & $\approx 2^{106.291}$\\
\hline
\end{tabular}
\caption{The cardinality of $|\near{\MF{2n}}|$ and $|\MFSP{2n}|$}
\label{table:NMFandMFSPPower}
\end{table}

%

\section{An approach to describe $\MFC{2n}$}\label{sec:mfcdescr}

To estimate the cardinalities of $\MFC{2n}$ and $\near{\MFC{2n}}$, we need a convenient representation of $\MFC{2n}$. Let us start with the property detecting functions from it, see, for instance,~\cite[Lemma 33]{CanteautEtAl2006}. It is also given in~\cite{Dillon1974} in terms of derivatives.
\begin{lemma}[see~\cite{CanteautEtAl2006}]\label{lemma:mfcdescription}
	A function $f: \z{2n} \to \z{}$ belongs to $\MFC{2n}$ $\iff$ there exists some $U \in \snk{2n}{n}$ such that $f$ is affine on each coset of $U$.
\end{lemma}

We introduce the following denotations. 
\begin{itemize}
\item Let $U_1, \ldots, U_{\snkp{2n}{n}}$ be all elements of $\snk{2n}{n}$, $U_1 = \xspace{n}$. 
\item Let $\MF{U_i}$ consist of all $f \in \bentset{2n}$ which are affine on each coset of $U_i$, $1 \leq i \leq \snkp{2n}{n}$.
\end{itemize}
According to Lemma~\ref{lemma:mfcdescription}, 
\begin{equation}\label{eq:mfcasunion}
	\MFC{2n} = \bigcup_{i = 1}^{\snkp{2n}{n}} \MF{U_i}.
\end{equation}
Moreover, these sets are connected through affine equivalence.
\begin{lemma}\label{lemma:mfdescrequivalence}
    Let linear $A: \z{2n} \to \z{2n}$ be invertible and $A^{-1}(U_1) = U_i$ for some $i$, $1 \leq i \leq \snkp{2n}{n}$. Then 
    $$
        \MF{U_i} = \MF{2n} \circ A = \{ f \circ A : f \in \MF{2n}\} 
    $$
    and $\near{f} \circ A = \near{f \circ A}$.
    In particular, $|\MF{U_i}| = |\MF{2n}|$ for any $i$ and $|\near{f}| = |\near{f \circ A}|$.
\end{lemma}
\begin{proof}
	It is clear that $\MF{2n} \circ A \subseteq \MF{U_i}$. Indeed, any function from $\MF{2n} \circ A$ is affine on each coset of $U_i$ since any $f(A(x))$ is affine on each coset of $A^{-1}(U_1) = U_i$, where $f \in \MF{2n}$. Similarly, any bent function $f \in \MF{U_i} \circ A^{-1}$ is affine on each coset of $U_1$, i.e. $f$ can be represented in the following form:
	$$
		f(x, y) = \iprod{x}{\sigma(y)} \oplus \psi(y), 
	$$
	where $x, y \in \z{n}$, $\sigma: \z{n} \to \z{n}$ and $\psi: \z{n} \to \z{}$. But it is well known that $f$ cannot be a bent function if $\sigma$ is not invertible, see, for instance,~\cite{CarletBook}. Hence, $\MF{U_i} = \MF{2n} \circ A$. 

    Let $U \in \asnk{2n}{n}$ and $f \in \MF{2n}$. It is clear that 
    $f \oplus \charfunc{U} \in \bentset{2n}$ $\iff$ the function
    \begin{equation}\label{eq:global:nearestAEq}
        f(A(x)) \oplus \charfunc{U}(A(x)) = f(A(x)) \oplus \charfunc{A^{-1}(U)}(x)
    \end{equation}
    is a bent function. This implies $\near{f \circ A} = \near{f} \circ A$. 
\end{proof}
Thus, the following upper bound for $|\MFC{2n}|$ is correct:
\begin{equation}\label{eq:mfcupperbound}
	|\MFC{2n}| \leq \snkp{2n}{n}|\MF{2n}|.
\end{equation}
Moreover, we will prove that it is asymptotically tight (Section~\ref{sec:BoundsMFC}). This bound is much more precise than the trivial 
$$
		|\MFC{2n}| \leq (2^{2n} - 2^0) \cdot \ldots \cdot (2^{2n} - 2^{2n - 1})|\MF{2n}|.
$$

We note that the functions belonging to only one of $\MF{U_1}, \ldots, \MF{U_{\snkp{2n}{n}}}$ will be mainly taken into account in Sections~\ref{sec:BoundsMFC} and \ref{sec:nearestAverageMFC}. See, for instance,  Propositions~\ref{state:MFUnique},  \ref{state:mcboundasympt} and Lemma~\ref{lemma:twoCosetSeries} that will allow us to omit other ones. Such functions are said to have \textit{the unique $\mathcal{M}$-subspace}. The set of them for some $f$ in $2n$ variables is defined as 
$$
	\mathcal{M}(f) = \{ U_i : f \in \MF{U_i}, 1 \leq i \leq \snkp{2n}{n} \}. 
$$
They were investigated, for instance, in~\cite{PolujanThesis2021, PasalicEtAl2024}. 

We will use the following criterion for detecting if $f \in \MF{2n} \cap \MF{U_i}$ (equivalently, $U_i \in \mathcal{M}(f)$), which is the direct consequence of Theorem~\ref{th:criterionMMFnew}. 
\begin{corollary}\label{cor:fullCosetSerie}
	Let $f_{\pi, \varphi} \in \MF{2n}$ and $U = \sconstr{n}{L}{R}{H}{\infomaps}$, where $L \in \snk{n}{k}$, $R \in \snk{n}{n - k}$ and $H: L \to \z{k}$ is linear. 
	Then $f_{\pi, \varphi} \in \MF{U}$ $\iff$ the following conditions are satisfied:
	\begin{enumerate}
		\item $\pi|_{a \oplus L}$ is affine and $\linsp{\pi(a \oplus L)} = R^{\bot}$ for any $a \in \z{n}$,
	    \item $x \in a \oplus L \mapsto \iprod{H(a \oplus x)}{\perk{\pi}{I}(x)} \oplus \varphi(x)$ is affine for any $a \in \z{n}$, where $I = \infomaps(\linsp{\pi(L)})$.
	\end{enumerate}
\end{corollary}
\begin{proof}
	Let $a \in \z{n}$ and $b \in \z{k}$. It is clear that any coset of $U$ is represented as $\{(\coords{H}{I}(a \oplus y) \oplus \coords{b}{I} \oplus z, y) : y \in a \oplus L, z \in R\}$.
	Also, let $\xi_b: x \in a \oplus L \mapsto \iprod{H(x \oplus a) \oplus b}{\perk{\pi}{I}(x)} \oplus \varphi(x)$. According to Theorem~\ref{th:criterionMMFnew}, $f_{\pi, \varphi} \in \MF{U}$ $\iff$ $\linsp{\pi(a \oplus L)} = R^{\bot}$, and $\xi_b$ is affine for each $a$ and $b$.  
	
	Let $\xi_{b}$ be affine for each $b$. Then $\xi_{b} \oplus \xi_{0}$ is affine as well, 
	$$
		\xi_{b}(x) \oplus \xi_{0}(x) = \iprod{b}{\perk{\pi}{I}(x)}, \ x \in a \oplus L.
	$$
	Hence, $\perk{\pi}{I}$ is affine on $a \oplus L$.  But $I$ is an information set of $\linsp{\pi(a \oplus L)} = \linsp{\pi(L)} = R^{\bot}$ and, therefore, of $\pi(a \oplus L) \in \asnk{n}{k}$. This means that $\perk{y}{\overline{I}} = B(\perk{y}{I})$ for all $y \in \pi(a \oplus L)$, where $B: \z{k} \to \z{n - k}$ is some affine function. Consequently, $\pi|_{a \oplus L}$ must be affine. 
    
	At the same time, the affinity of $\pi|_{a \oplus L}$ implies that $\xi_{b}(x)$ is affine for each $b$ $\iff$ $\xi_{0}$ is affine. Thus, $f_{\pi, \varphi} \in \MF{U}$ $\iff$ $\linsp{\pi(a \oplus L)} = R^{\bot}$, and $\pi|_{a \oplus L}$, $\xi_{0}$ are affine. 
\end{proof}

\section{Bounds for $|\MFC{2n}|$}\label{sec:BoundsMFC}

In this section we propose bounds for $|\MFC{2n}|$ and prove that they are asymptotically tight. 
We start with the main advantage of Corollary~\ref{cor:fullCosetSerie}, which is the possibility to construct all functions from $\MF{2n} \cap \MF{U}$ for any $U \in \asnk{2n}{n}$. 

\begin{lemma}\label{lemma:MFUintersection}
	Let $U \in \snk{2n}{n}$ and $\dim (U \cap \xspace{n}) = k < n$. Then
	$$
		|\MF{2n} \cap \MF{U}| = 2^k! \, 2^{(2n - 2k + 1)2^k} \prod_{i = 0}^{n - k - 1}(2^{n - k} - 2^i)^{2^k} .
	$$
\end{lemma}
\begin{proof}
	Let us construct all suitable $f_{\pi, \varphi} \in \MF{2n}$. 
    We can represent $U$ as $\sconstr{n}{L}{R}{H}{\infomaps}$, where $R = U \cap \xspace{n}$, $L \in \snk{n}{n - k}$ and $H: L \to \z{n - k}$ is linear, see Proposition~\ref{state:subspaceRepresentation}.
    Let $a_1 \oplus L, \ldots, a_{2^{k}} \oplus L$ and $b_1 \oplus R^{\bot}, \ldots, b_{2^{k}} \oplus R^{\bot}$ be all distinct cosets of $L$ and $R^{\bot}$, $a_1, \ldots, a_{2^k}, b_1, \ldots, b_{2^k} \in \z{n}$. 
	
	According to the conditions of Corollary~\ref{cor:fullCosetSerie}, $f_{\pi, \varphi} \in \MF{U}$ $\iff$ we choose a permutation $\pi$ and $\varphi$ in the following way.
		
    1. $\pi$ transforms $a_1 \oplus L, \ldots, a_{2^k} \oplus L$ to $b_1 \oplus R^{\bot}, \ldots, b_{2^{k}} \oplus R^{\bot}$ in any order. Since $\pi$ must be invertible, we have $2^k!$ possibilities to choose $\pi(a_1 \oplus L), \ldots, \pi(a_{2^k} \oplus L)$.
		
    2. $\pi$ is affine on each of $a_1 \oplus L, \ldots, a_{2^k} \oplus L$ and invertible. The number of such transformations is $2^{n - k}(2^{n - k} - 2^0) \cdot \ldots \cdot (2^{n- k} - 2^{n - k - 1})$ for each of $2^k$ cosets.
    Overall, we can choose $\pi$ in $2^k! \, 2^{(n - k) 2^k} \prod_{i = 0}^{n - k - 1}(2^{n - k} - 2^i)^{2^k}$ ways to satisfy the first condition of Corollary~\ref{cor:fullCosetSerie}.
		
    3. For each $1 \leq i \leq 2^k$, the second condition of Corollary~\ref{cor:fullCosetSerie} is satisfied $\iff$ the given $H$ and already chosen $\pi$ determine $\varphi|_{a_i \oplus L}$ up to an affine function, i.e. there are $2^{n - k + 1}$ ways to choose $\varphi|_{a_i \oplus L}$ and $2^{(n - k + 1)2^k}$ ways to choose $\varphi$.
    The result is the product of ways to choose $\pi$ and then $\varphi$. 
\end{proof}

\begin{proposition}\label{state:MFUnique}
	There are at least $|\MF{2n}| - \mfcbound{2n}$ bent functions $f \in \MF{2n}$ with $|\mathcal{M}(f)| = 1$, 
	where $\mfcbound{2n}$ is equal to 
	\begin{equation}\label{eq:mfcboundex}
		\sum_{k = 0}^{n - 1} {(\snkp{n}{k})}^2 2^{(n - k)^2} 2^k! \, 2^{(2n - 2k + 1)2^k} \prod_{i = 0}^{n - k - 1}(2^{n - k} - 2^i)^{2^k}.
	\end{equation}
\end{proposition}
\begin{proof}
	The set of bent functions belonging to $\MF{2n}$ with the unique $\mathcal{M}$-subspace is
	$
		\MF{2n} \setminus ( \MF{U_2} \cup \ldots \cup  \MF{U_{\snkp{2n}{n}}}),
	$
	see Section~\ref{sec:mfcdescr}. The lower bound for its cardinality is
	$
		|\MF{2n}| - \sum_{i = 2}^{\snkp{2n}{n}} |\MF{2n} \cap \MF{U_i}|.
	$
	Thus, it is sufficient to prove that 
	\begin{equation}\label{eq:intersectionMFSum}
		\mfcbound{2n} = \sum_{i = 2}^{\snkp{2n}{n}} |\MF{2n} \cap \MF{U_i}|.
	\end{equation}
	Lemma~\ref{lemma:MFUintersection} gives us $|\MF{2n} \cap \MF{U_i}|$ depending on $\dim (U_i \cap \xspace{n}) = k$. 
	Due to Proposition~\ref{state:subspaceRepresentation}, the number of $U = \sconstr{n}{L}{R}{H}{\infomaps}\in \snk{2n}{n}$ such that $\dim (U \cap \xspace{n}) = \dim R = k$ is equal to $2^{(n - k)^2} \snkp{n}{k} \snkp{n}{n - k}$. We can choose any $R \in \snk{n}{k}$, any $L \in \snk{n}{n - k}$ and any linear $H: L \to \z{n - k}$. Applying Lemma~\ref{lemma:MFUintersection}, we obtain that~(\ref{eq:mfcboundex})
	is equal to $\sum_{i = 2}^{\snkp{2n}{n}} |\MF{2n} \cap \MF{U_i}|$.
\end{proof}
This can be rewritten in the following way.
\begin{corollary}
	The expected value of $|\mathcal{M}(f)|$ for a random $f \in \MF{2n}$ is equal to the following:
	$$
		\frac{1}{|\MF{2n}|}\sum_{f \in \MF{2n}} |\mathcal{M}(f)| = 1 + \frac{\mfcbound{2n}}{|\MF{2n}|}.
	$$
	Its values for small $n$ can be found in Table~\ref{table:expectedMf}.
\end{corollary}
\begin{table}[!t]
\renewcommand{\arraystretch}{1.3}
\centering
\begin{tabular}{cc||cc}
\hline
$2n$ & Expected $|\mathcal{M}(f)|$ & $2n$ & Expected $|\mathcal{M}(f)|$ \\
\hline
2 & 3 & 10 & $1 + 2^{-46.501079} $ \\
4 & 15 & 12 & $1 + 2^{-133.377320}$ \\
6 & 8.6 & 14 & $1 + 2^{-341.189209}$ \\
8  &  $1 + 2^{-10.349626}$ & 16 & $1 + 2^{-822.845858}$ \\
\hline
\end{tabular}
\caption{The approximated expected value of $|\mathcal{M}(f)|$, $f \in \MF{2n}$}
\label{table:expectedMf}
\end{table}
\begin{proof}
	Indeed, similarly to~(\ref{eq:sumPi}) in the proof of Lemma~\ref{lemma:sumPi},
	$
		\sum_{f \in \MF{2n}} |\mathcal{M}(f)| = |\MF{2n} \cap \MF{2n}| + \sum_{i = 2}^{\snkp{2n}{n}} |\MF{2n} \cap \MF{U_i}|
	$
	since each $\MF{2n} \cap \MF{U_i}$ is exactly the set $
	 	\{f \in \MF{2n} : U_i \in \mathcal{M}(f)\}
	$. The equality~(\ref{eq:intersectionMFSum}) completes the proof.
\end{proof}

\begin{theorem}\label{th:MFCLowerBound}
	For $\mfcbound{2n}$ defined in~(\ref{eq:mfcboundex}) the following holds:
	$$\snkp{2n}{n}|\MF{2n}| - \snkp{2n}{n} \mfcbound{2n} \leq |\MFC{2n}| \leq \snkp{2n}{n}|\MF{2n}|.$$  
\end{theorem}
\begin{proof}
	The proof follows from Proposition~\ref{state:MFUnique} and~(\ref{eq:mfcasunion}). Indeed,
	$
		\MFC{2n} = \cup_{i = 1}^{\snkp{2n}{n}} \MF{U_i}
	$
	and each $\MF{U_i}$ contains at least $|\MF{2n}| - \mfcbound{2n}$ bent functions that do not belong to $\MF{U_j}$ for $i \neq j$.
\end{proof}

\begin{remark}
	We note that the lower bound from Theorem~\ref{th:MFCLowerBound} takes into account only $f \in \MFC{2n}$ with $|\mathcal{M}(f)| = 1$. 
\end{remark}

Table~\ref{table:mfcestimations} contains the values of the bounds for $|\MFC{2n}|$ from Theorem~\ref{th:MFCLowerBound}, where $n$ is small. 
Note that only $|\MFC{2}| = 8$, $|\MFC{4}| = 896$ and $|\MFC{6}| = 5425430528$ are known, they consist of all bent functions.
Theorem~\ref{th:MFCLowerBound} gives us $|\MFC{8}| \approx 2^{77.865}$, which is better than known $|\MFC{8}| < 2^{81.38}$~\cite{LangevinLeander2011}.

\begin{table}[!t]
\renewcommand{\arraystretch}{1.3}
\centering
\begin{tabular}{ccc||ccc}
\hline
$2n$ & Lower & Upper & $2n$ & Lower & Upper \\
\hline
2 & $< 0$ & 4.584963 & 10 & 176.365947 & 176.365947\\
4 & $< 0$ & 13.714246 & 12 & 397.742211 & 397.742211\\
6 & $< 0$ & 33.745257 & 14 & 894.931155 & 894.931155\\
8 & 77.864341 & 77.865447 & 16 & 2005.776948 & 2005.776948\\
\hline
\end{tabular}
\caption{Bounds for $\log_2 |\MFC{2n}|$ from Theorem~\ref{th:MFCLowerBound}}
\label{table:mfcestimations}
\end{table}

\subsection{The asymptotics of $|\MFC{2n}|$ and additional bounds}\label{sec:boundsasympt}

Let us estimate $\mfcbound{2n}$ and apply the results to~$|\MFC{2n}|$.
\begin{proposition}\label{state:mcboundasympt}
	  $(2^{n} - 1)^2 2^{3 \cdot 2^{n - 1} + 1} 2^{n - 1}! \leq \mfcbound{2n}$. Also, 
	\begin{equation}\label{eq:mfcboundlowerupper}
        \mfcbound{2n} < 2^{3 \cdot 2^{n - 1} + 2n + 1} 2^{n - 1}! \text{  for any  } n \geq 7.
	\end{equation}
\end{proposition}
\begin{proof}
	According to~(\ref{eq:mfcboundex}), $\mfcbound{2n} = \sum_{k = 0}^{n - 1} T_k$, where
	\begin{equation}\label{eq:estimatingSum}
		T_k =  {\snkp{n}{k}}^2 2^{(n - k)^2} 2^k! \, 2^{(2n - 2k + 1)2^k} \prod_{i = 0}^{n - k - 1}(2^{n - k} - 2^i)^{2^k}.
	\end{equation}
	It means that
	$
		T_{n - 1} = 2^{3 \cdot 2^{n - 1} + 1} (2^{n} - 1)^2 2^{n - 1}!,
	$
	i.e. $T_{n - 1} = U_n - o(U_n)$, where $U_n = 2^{3 \cdot 2^{n - 1} + 2n + 1} 2^{n - 1}!$. This proves the lower bound.
	Thus, it is sufficient to prove that $T_{0} + \ldots + T_{n - 2} = o(U_{n})$ and $T_{0} + \ldots + T_{n - 2} \leq R_{n}$, where
	\begin{equation*}
		R_{n} = 2^{3 \cdot 2^{n - 1} + n + 1} 2^{n - 1}! = 2^{- n } U_n
		 < U_n - T_{n - 1}.
	\end{equation*}
	Let $k \leq n - 2$. First, 	
	\begin{equation}\label{eq:snkpapprox}
		\snkp{n}{k} = 
		2^{k(n - k)}\prod_{i = 0}^{k - 1}\frac{1 - 2^{i - n}}{1 - 2^{i - k}} \leq 2^{k(n - k ) + k}
	\end{equation}
	since $1 - 2^{i - n} \leq 1$ and $1 - 2^{i - k} \geq 1/2$. Therefore,
	\begin{equation*}
		\snkp{n}{k}^2 2^{(n - k)^2} = \snkp{n}{k} \snkp{n}{n - k} 2^{(n - k)^2} 
        \leq 2^{2k(n - k) + n + (n - k)^2} = 2^{(n - k)(n + k) + n} \leq 2^{n^2 + n}.
	\end{equation*}
	At the same time $(2^{n - k} - 2^0) \cdot \ldots \cdot (2^{n - k} - 2^{n - k - 1}) \leq 2^{(n - k)^2 - 1}$, which implies that
	\begin{equation*}
		2^{(2n - 2k + 1)2^k} \prod_{i = 0}^{n - k - 1}(2^{n - k} - 2^i)^{2^k} \leq 2^{((n - k)^2 + 2(n - k))2^k} 
        = 2^{\frac{(n - k)(n - k +2)}{2^{n - k}} 2^n} \leq 2^{2^{n + 1}}.
	\end{equation*}
	Also,
	\begin{equation}\label{eq:factorialdiv}
		\frac{2^k!}{2^{k + 1}!} \leq {(2^{k})}^{-2^{k}} \text{ and } \frac{2^k!}{2^{n - 1}!} \leq 2^{-(n - 2)2^{n - 2}}.
	\end{equation} 
	As a result, we obtain that 
	\begin{equation}\label{eq:TkInequality}
		T_k \leq 2^{n^2 + n + 2^{n + 1}} 2^k! \leq 2^{n^2 - n - 1 + 2^{n - 1} - (n - 2)2^{n - 2}} U_n,
	\end{equation} 
	which directly implies that $T_0 + \ldots + T_{n - 2} = o(U_n)$. Similarly,
	$$
		T_k \leq 2^{n^2 - 1 - (n - 4)2^{n - 2}} R_n \text{ and }
	$$
	$$
		n^2 - 1 - (n - 4)2^{n - 2} \leq -48 \text{ for any } n \geq 7.
	$$
	Due to~(\ref{eq:factorialdiv}) and (\ref{eq:TkInequality}), $T_0 + \ldots + T_{n - 2} < R_n$ holds.
\end{proof}

Table~\ref{table:summusubsp} shows us rounded values of $\mfcbound{2n}$ for small $n$ and its upper bound~(\ref{eq:mfcboundlowerupper}). 
\begin{remark}
    Proposition~\ref{state:mcboundasympt} is correct for $n \geq 5$, see Table~\ref{table:summusubsp}. 
\end{remark}

\begin{table}[!t]
\renewcommand{\arraystretch}{1.3}
\caption{Rounded values of $\log_2 \mfcbound{2n}$ and its upper bound}
\label{table:summusubsp}
\centering
\begin{tabular}{cccccc}
\hline
$2n$ & $\log_2 \mfcbound{2n}$ &  $\log_2$ of~(\ref{eq:mfcboundlowerupper}) & $\log_2 (\snkp{2n}{n} \mfcbound{2n})$ & $\log_2 |\MF{2n}|$\\
\hline
8 &  49.900515 &  48.299208 & 67.515821 & 60.250140\\
10 &  103.162185 &  103.250140 & 129.864868 & 149.663264\\
12 &  226.617823 &  226.663264 & 264.364890 & 359.995144\\
14 &  502.972513 &  502.995144 & 553.741947 & 844.161722\\
16 &  1117.150429 &  1117.161722 & 1182.931089 & 1939.996287\\
\hline
\end{tabular}
\end{table}

\begin{corollary}\label{cor:MFCLowerBoundAsympt}
	$|\MFC{2n}| = \snkp{2n}{n}|\MF{2n}| - o(2^n!)$.  
\end{corollary}
The proof directly follows from Theorem~\ref{th:MFCLowerBound}, Proposition~\ref{state:mcboundasympt} and~(\ref{eq:factorialdiv}). Also, Proposition~\ref{state:mcboundasympt} gives us the bound
\begin{equation}\label{eq:mfclowerboundbyestimation}
    |\MFC{2n}| > \snkp{2n}{n} (2^{2^n} 2^n! - 2^{3 \cdot 2^{n - 1} + 2n + 1} 2^{n - 1}!), \ n \geq 5
\end{equation}
as well as a more accurate asymptotic expression in the case of applying it together with Stirling's formula.

\section{An upper bound for $|\near{\MFC{2n}}|$}\label{sec:nearestAverageMFC}

We will estimate the $|\near{\MFC{2n}}|$ using 
$$
    \nearav{\MFC{2n}} = \frac{1}{|\MFC{2n}|} \sum_{f \in \MFC{2n}} |\near{f}|.
$$
The main idea is the following lemma. 
\begin{lemma}\label{lemma:twoCosetSeries}
    Let $f \in \MF{U} \cap \MF{U'}$ for distinct $U, U' \in \snk{2n}{n}$. Then $|\near{f}| \geq 2^{2n + 2} - 2^{n + 3}$. In particular, $|\near{f}| > \nearav{\MF{2n}}$ if $n \geq 5$.
\end{lemma}
\begin{proof}
    Due to Corollary~\ref{cor:AverageNumbUpperBound}, $2^{2n + 2} - 2^{n + 3} > \nearav{\MF{2n}}$ if $n \geq 5$.
    
    According to Lemma~\ref{lemma:mfdescrequivalence}, we can assume that $U' = \xspace{n}$, i.e. $\MF{U'} = \MF{2n}$ and $f = f_{\pi, \varphi} \in \MF{2n} \cap \MF{U}$. 
	Since $U \neq \xspace{n}$, $\dim R \leq n - 1$ holds, where $R = U \cap \xspace{n}$. Also, Proposition~\ref{state:subspaceRepresentation} and Corollary~\ref{cor:fullCosetSerie} give us that $U = \sconstr{n}{L}{R}{H}{\infomaps}$, where 
    $\linsp{\pi(L)} = R^{\bot}$ and $L \in \snk{n}{n - \dim R}$.
	    
	Let $\dim R = n - 1$, i.e. $\dim L = 1$. According to Corollary~\ref{cor:fullCosetSerie}, $\linsp{\pi(a \oplus L)} = R^{\bot}$ for all $a \in \z{n}$. 
    It implies that $\pi(a \oplus L) \cup \pi(a' \oplus L) \in \asnk{n}{2}$ for any two distinct cosets $a \oplus L$ and $a' \oplus L$, $a' \in \z{n}$. Thus,  $(a \oplus L) \cup (a' \oplus L) \in \subsp{2}{\pi}$. Therefore, 
    $
        |\subsp{2}{\pi}| \geq \frac{2^{n - 1} \cdot (2^{n - 1} - 1)}{2} = 2^{2n - 3} - 2^{n - 2}
    $
    since there are $2^{n - 1}$ distinct cosets of $L$. By Theorem~\ref{theorem:criterionMMFdim2}, $\near{f_{\pi, \varphi}}$ contains at least $\lbound{2n} + 2^5 |\subsp{2}{\pi}| \geq 2^{2n + 2} - 2^{n + 3}$ elements. 

    Let $\dim R \leq n - 2$. We recall that $\near{f_{\pi, \varphi}}$ contains $\lbound{2n}$-element subsets $G$ and $G'$ relatively to $f_{\pi, \varphi} \in \MF{U}$ and $f_{\pi, \varphi} \in \MF{2n}$, see Proposition~\ref{state:mcfarlandcase} and Lemma~\ref{lemma:mfdescrequivalence}. However, we must take into account that some elements of $G$ and $G'$ may coincide. Let us find an upper bound for $|G \cap G'|$ since $|G \cup G'| = 2 \lbound{2n} - |G \cap G'|$. 
	Suppose that some $g \in G \cap G'$ is constructed using $S \in \asnk{2n}{n}$, i.e. $S = \sconstr{n}{L'}{R'}{H'}{\infomaps}$ due to Theorem~\ref{th:criterionMMFnew}. We need to estimate the number of such $S$.
    
    Proposition~\ref{state:mcfarlandcase} and Lemma~\ref{lemma:mfdescrequivalence} give us that $\dim (\linsp{S} \cap \xspace{n}) \geq n - 1$ and $\dim (\linsp{S} \cap U) \geq n - 1$. 
    Hence, $\dim (\linsp{S} \cap U \cap \xspace{n}) \geq n - 2$. At the same time, $\dim R = \dim (U \cap \xspace{n}) \leq n - 2$ which is possible only if $R = \linsp{S} \cap U \cap \xspace{n}$. Consequently, $R \subset \linsp{S}$ and $\dim R = n - 2$. Due to Proposition~\ref{state:subspaceRepresentation}, $R \subset R'$ since $R, R' \subset \xspace{n}$ and $R' = \linsp{S} \cap \xspace{n}$.
	
	Next, $\dim (\linsp{S} \cap \xspace{n}) = n - 1$. Indeed, it is shown that $\dim (\linsp{S} \cap \xspace{n}) \geq n - 1$, but $\dim (\linsp{S} \cap \xspace{n}) = n$ implies $\linsp{S} = \xspace{n}$ which together with $\dim (\linsp{S} \cap U) \geq n - 1$ contradicts $\dim R = n - 2$.
    
    According to Proposition~\ref{state:mcfarlandcase}, $L' \in \subsp{1}{\pi}$. To calculate the number of all possible $S = \sconstr{n}{L'}{R'}{H'}{\infomaps}$, we need to calculate the number of possible $L' \in \subsp{1}{\pi}$ such that $R \subset R' = \linsp{\pi(L')}^{\bot}$ and $H'$. In the case of $\dim L' = 1$ the second condition of Theorem~\ref{th:criterionMMFnew} is satisfied for any affine $H'$, there are $2^{1\cdot1 +1} = 4$ of them.
    
    Next, $R'^{\bot} = \linsp{\pi(L')} = \linsp{\pi(\{x, y\})\}}$ for all distinct $x, y \in \z{n}$. Since $\pi(\{x, y\})$ is any element of $\asnk{n}{1}$, each distinct $R'$ will appear $2^{n - 1}$ times, i.e. for each coset of $R'^{\bot}$. At the same time, we are interested only in $R' \supset R$.
    There are $3$ of such $R'$. We can union $(n - 2)$-dimensional $R$ with any of its coset that is not equal to $R$ to obtain $(n - 1)$-dimensional $R'$. 
    Thus,
    $
        |G \cap G'| \leq  4 \cdot 3 \cdot 2^{n -1} = 3 \cdot 2^{n + 1},
    $
    and $|G \cup G'| \geq 2 \lbound{2n} - 3 \cdot 2^{n + 1} = 2^{2n + 2} - 2^{n + 3}$. 
\end{proof}
This allows us to obtain the following upper bound. 

\begin{theorem}\label{th:MFCompletedNearAverage}
    Let $n \geq 5$. Then $\nearav{\MFC{2n}} < \nearav{\MF{2n}}$. Moreover, $\nearav{\MFC{2n}} = \nearav{\MF{2n}} - o(1)$.
\end{theorem}
\begin{proof}    
    In terms of Section~\ref{sec:mfcdescr}, let $\MF{U_i}^*$ be all functions from $\MF{U_i}$ that do not belong to $\MF{U_j}$ for any $1 \leq j < i$. Then~(\ref{eq:mfcasunion}) implies that 
    $
        |\MFC{2n}| = \sum_{i = 1}^{m}|\MF{U_i}^*|,
    $
    i.e. $\{ \MF{U_1}^*, \ldots, \MF{U_m}^*\}$ is the partition of $\MFC{2n}$, $m = \snkp{2n}{n}$. Therefore, 
    \begin{equation}\label{eq:MFCompletedPart}
        \nearav{\MFC{2n}}|\MFC{2n}| = \sum_{i = 1}^{m} \sum_{f \in \MF{U_i}^*} |\near{f}|.
    \end{equation}
    At the same time, Lemma~\ref{lemma:mfdescrequivalence} guarantees that
    \begin{equation*}
         \nearav{\MF{2n}}|\MF{2n}| = \sum_{f \in \MF{U_i}^*} |\near{f}| + \sum_{f \in \MF{U_i} \setminus \MF{U_i}^*} |\near{f}|. 
    \end{equation*}
%
    But for each $f \in \MF{U_i} \setminus \MF{U_i}^*$ the condition of Lemma~\ref{lemma:twoCosetSeries} is satisfied. Hence,   Lemma~\ref{lemma:twoCosetSeries} implies the following inequality:
    \begin{equation}\label{eq:MFStarIneq}
        \sum_{f \in \MF{U_i} \setminus \MF{U_i}^*} |\near{f}| >^* \nearav{\MF{2n}} \big(|\MF{U_i}| - |\MF{U_i}^*|\big).
    \end{equation}
    The used $>^*$ means that we must use $\geq$ if $\MF{U_i} = \MF{U_i}^*$. At the same time, there is some $i$ with $\MF{U_i} \neq \MF{U_i}^*$. For instance, the bent function $\iprod{x}{y}$ belongs to both $\MF{\z{n} \times \{ 0 \in \z{n}\}}$ and $\MF{ \{ 0 \in \z{n}\} \times \z{n}}$. 
    Next, 
    (\ref{eq:MFStarIneq}) and the equality above imply
    \begin{multline*}
        \sum_{f \in \MF{U_i}^*} |\near{f}| <^* \nearav{\MF{2n}}|\MF{2n}| - \nearav{\MF{2n}} |\MF{U_i}| \\+ \nearav{\MF{2n}} |\MF{U_i}^*| = \nearav{\MF{2n}} |\MF{U_i}^*|. 
    \end{multline*}
    Substituting this inequality to~(\ref{eq:MFCompletedPart}), we obtain that
    \begin{equation*}
        \nearav{\MFC{2n}}|\MFC{2n}| 
        < \sum_{i = 1}^{m} \nearav{\MF{2n}} |\MF{U_i}^*| =  \nearav{\MF{2n}} |\MFC{2n}|.
    \end{equation*}
    
    Let us provide a lower bound for $\nearav{\MFC{2n}}$. First of all, there is the bound $|\near{f}| \leq 2^n(2^1 + 1) \cdot \ldots \cdot (2^n + 1) \leq 2^{n (n  + 5)/ 2}$, see~\cite{Kolomeec2017}. Also, $|\MF{U_i} \setminus \MF{U_i}^*| \leq \mfcbound{2n}$ due to Proposition~\ref{state:MFUnique} and Lemma~\ref{lemma:mfdescrequivalence}. Thus,
    $$
        \sum_{f \in \MF{U_i} \setminus \MF{U_i}^*} |\near{f}| \leq 2^{n (n  + 5)/ 2} \mfcbound{2n}
    $$
    which can be used similarly to~(\ref{eq:MFStarIneq}):
    $$
        \sum_{f \in \MF{U_i}^*} |\near{f}| \geq  \nearav{\MF{2n}}|\MF{2n}| - 2^{n (n  + 5)/ 2} \mfcbound{2n}. 
    $$
    Substituting this inequality to~(\ref{eq:MFCompletedPart}), we obtain that
    \begin{equation*}
        \nearav{\MFC{2n}}|\MFC{2n}| 
        \geq  \snkp{2n}{n} \nearav{\MF{2n}}|\MF{2n}| - \snkp{2n}{n} 2^{n (n  + 5)/ 2} \mfcbound{2n}.
    \end{equation*}
    Taking into account~(\ref{eq:mfcupperbound}),  Proposition~\ref{state:mcboundasympt} and (\ref{eq:factorialdiv}), 
    this gives us $\nearav{\MFC{2n}} 
    \geq \nearav{\MF{2n}} - o(1)$.
\end{proof}

Let us estimate the cardinality of $\near{\MFC{2n}}$ and $\MFCSP{2n}$.
\begin{corollary}\label{cor:NearCompletedUpperBound}
    Let $n \geq 5$. Then $\near{\MFC{2n}} < (\nearav{\MF{2n}} - \lbound{2n}) |\MFC{2n}|$. In particular, 
    $
        \near{\MFC{2n}} < (\frac{4}{3} 2^{2n} + 29) |\MFC{2n}|
    $ and $|\MFCSP{2n}| < (\frac{4}{3} 2^{2n} + 30)|\MFC{2n}|$.
\end{corollary}
\begin{proof}
    It is clear that
    $
        \near{\MFC{2n}} \leq \sum_{f \in \MFC{2n}}\big( |\near{f}| - \lbound{2n} \big) =  \nearav{\MFC{2n}} |\MFC{2n}| - \lbound{2n} |\MFC{2n}|  
    $
    since each $\near{f}$ contains at least $\lbound{2n}$ bent functions from $\MFC{2n}$, see Proposition~\ref{state:mcfarlandcase} and Lemma~\ref{lemma:mfdescrequivalence}. Theorem~\ref{th:MFCompletedNearAverage} and Corollary~\ref{cor:AverageNumbUpperBound} complete the proof.
\end{proof}
\begin{remark}\label{remark:NearMFC}
Due to Theorem~\ref{th:MFCompletedNearAverage},
$\nearav{\MFC{2n}}$, $|\near{\MFC{2n}}|$ and $|\MFCSP{2n}|$ can be estimated 
more precisely using Corollary~\ref{cor:AverageNumbSimplified}. 
\end{remark} 

The estimations of $|\MFC{2n}|$ from  Theorem~\ref{th:MFCLowerBound} can also be used.
Note that providing a lower bound for $|\near{\MFC{2n}}|$ and $|\MFCSP{2n}|$ looks more difficult since even finding $\near{f} \cap \MFC{2n}$ is not an easy problem. 

\section{Conclusion}

We have analyzed the affinity of random $f \in \MF{2n}$ and $g \in \MFC{2n}$ on $n$-dimensional affine subspaces of $\z{2n}$ and have obtained that the properties of $\MF{2n}$ and $\MFC{2n}$ are similar.
The results have allowed us to establish $|\MFSP{2n}|$ in a quite constructive way, precisely enough estimate $|\MFC{2n}|$ and obtain the upper bound for $|\MFCSP{2n}|$. These also imply certain metric properties of $\MF{2n}$ and $\MFC{2n}$. The missed thing is a lower bound for $|\MFCSP{2n}|$ which is a topic for further research.

\bigskip

\noindent {\normalsize \textbf{Acknowledgements} 
The work is supported by the Mathematical Center in Akademgorodok under the agreement No. 075--15--2022--282 with the Ministry of Science and Higher Education of the Russian Federation.}

\end{document}